\documentclass{llncs}
\usepackage[normalem]{ulem}
\usepackage[utf8]{inputenc}
\usepackage{amsmath}
\usepackage{amsfonts}
\usepackage{amssymb}
\usepackage{latexsym}
\usepackage{gastex}
\usepackage{graphicx}
\usepackage{marvosym}
\usepackage{multirow}
\usepackage{comment}
\usepackage{changebar}
\usepackage{ulem}
\def\fract#1/#2{\leavevmode
 \kern.1em \raise .5ex \hbox{\the\scriptfont0 #1}%
 \kern-.1em $/$%
 \kern-.15em \lower .25ex \hbox{\the\scriptfont0 #2}%
}
\def\abs#1{\ensuremath{\lvert #1\rvert}} 

\makeatletter
\DeclareRobustCommand\sfrac[1]{\@ifnextchar/{\@sfrac{#1}}%
                                            {\@sfrac{#1}/}}
\def\@sfrac#1/#2{\leavevmode\scalebox{.9}{\kern.1em\raise.5ex
         \hbox{$\m@th\mbox{\fontsize\sf@size\z@
                           \selectfont#1}$}\kern-.1em
         /\kern-.15em\lower.25ex
          \hbox{$\m@th\mbox{\fontsize\sf@size\z@
                            \selectfont#2}$}}}
\DeclareRobustCommand\numfrac[1]{\@ifnextchar/{\@numfrac{#1}}%
                                            {\@numfrac{#1}}}
\def\@numfrac#1{\leavevmode \hbox{$\m@th\mbox{\fontsize\sf@size\z@
                           \selectfont#1}$}}
\makeatother
\newcommand{\nat}{\mathbb N}
\newcommand{\tuple}[1]{\langle #1 \rangle}
\newcommand{\dist}{{\cal D}}

\newcommand{\q}{\hat{ q}}
\newcommand{\p}{\hat{ p}}
\newcommand{\M}{{\cal M}}
\newcommand{\N}{{\cal N}}
\newcommand{\B}{{\cal B}}

\newcommand{\A}{{\cal A}}
\newcommand{\F}{{\cal F}}
\newcommand{\LL}{{\cal L}}

\newcommand{\Supp}{{\sf Supp}}
\newcommand{\Pref}{{\sf Pref}}
\newcommand{\Plays}{{\sf Play}}

\newcommand{\Last}{{\sf Last}}
\newcommand{\sink}{{\sf sink}}
\newcommand{\Outcomes}{\mathit{Outcomes}}
\newcommand{\fsum}{\mathit{sum}}
\newcommand{\fmax}{\mathit{max}}
\newcommand{\win}[2]{\langle \! \langle 1 \rangle \! \rangle_{\mathit{#2}}^{\mathit{#1}}}
\newcommand{\winsure}[1]{\langle \! \langle 1 \rangle \! \rangle_{\mathit{sure}}^{\mathit{#1}}}
\newcommand{\winas}[1]{\langle \! \langle 1 \rangle \! \rangle_{\mathit{almost}}^{\mathit{#1}}}
\newcommand{\winlim}[1]{\langle \! \langle 1 \rangle \! \rangle_{\mathit{limit}}^{\mathit{#1}}}
\newcommand{\Bool}{{\sf B}^+}
\newcommand{\true}{{\sf true}}
\newcommand{\false}{{\sf false}}
\newcommand{\init}{{\sf init}}
\newcommand{\Act}{{\sf A}}
\newcommand{\Pre}{{\sf Pre}}
\newcommand{\post}{{\sf post}}
\newcommand{\mem}{{\sf Mem}}

\title{Limit Synchronization in \\ Markov Decision Processes\thanks{This work 
has been partly supported by the Belgian Fonds National de la Recherche 
Scientifique (FNRS).}}
\author{Laurent Doyen \inst{1}
\and Thierry Massart \inst{2} \and
Mahsa Shirmohammadi \inst{1,2} }

\institute{LSV, ENS Cachan \& CNRS, France 
\and Universit\'e Libre de Bruxelles, Belgium 
}

\begin{document}
\sloppy

\maketitle 

\pagestyle{plain}
\begin{abstract}
Markov decision processes (MDP) are finite-state systems with both
strategic and probabilistic choices. After fixing a strategy, an MDP
produces a sequence of probability distributions over states. 
The
  sequence is eventually synchronizing if the probability mass accumulates in a
  single state, possibly in the limit. Precisely, for $0 \leq p \leq 1$ the
  sequence is $p$-synchronizing if a probability distribution in the sequence  
  assigns probability at least~$p$ to some state, and we
 distinguish three synchronization modes:
$(i)$ \emph{sure} winning if there exists a strategy that produces a $1$-synchronizing sequence; 
$(ii)$ \emph{almost-sure} winning if there exists a strategy that produces a sequence that is, 
for all $\epsilon > 0$, a (1-$\epsilon$)-synchronizing sequence; 
$(iii)$ \emph{limit-sure} winning if for all $\epsilon > 0$, there exists a strategy that
produces a (1-$\epsilon$)-synchronizing sequence.
We consider the problem of deciding whether an MDP is sure,
almost-sure, or limit-sure winning, and we establish the decidability and
optimal complexity for all modes, as well as the memory requirements
for winning strategies.  Our main contributions are as follows: 
(a) for each winning modes we present characterizations
that give a PSPACE complexity for the decision problems, 
and we establish matching PSPACE lower bounds; 
(b) we show that for sure winning strategies, exponential memory is sufficient and may be necessary,
and that in general infinite memory is necessary for almost-sure winning,
and unbounded memory is necessary for limit-sure winning;
(c) along with our results, we establish new complexity results for
alternating finite automata over a one-letter alphabet.
\end{abstract}
\section{Introduction}

Markov decision processes (MDP) are finite-state stochastic processes
used in the design of systems that exhibit both controllable and
stochastic behavior, such as in planning, randomized algorithms, and
communication protocols~\cite{AspnesH90,FokkinkP06,BBS06}. 
The controllable choices along the execution are fixed by a strategy,
and the stochastic choices describe the system response.
When a strategy is fixed in an MDP, the \emph{symbolic} semantics 
is a sequence of probability distributions over states of the MDP, which
differs from the \emph{traditional} semantics where a probability
measure is considered over sets of sequences of states. 
This semantics is adequate in many
applications, such as systems biology, sensor networks, robot
planning, etc.~\cite{HMW09,BBMR08}, where the system
consists of several copies of the same process (molecules, sensors,
robots, etc.), and the relevant information along the execution of the
system is 
the number of processes in each state, or the relative frequency (i.e., the probability)
of each state.
In recent works, the verification of quantitative properties of the symbolic
semantics was shown undecidable~\cite{KVAK10}.
Decidability is obtained for special subclasses~\cite{CKVAK11}, 
or through approximations~\cite{AAGT12}.

In this paper, we consider a general class of strategies that select
actions depending on the full history of the system execution.  In the
context of several identical processes, the same strategy is used in
every process, but the internal state of each process need not be the
same along the execution, since probabilistic transitions may have
different outcome in each process. Therefore, the execution of the
system is best described by the sequence of probability distributions
over states along the execution.  Previously, the special case of
word-strategies have been considered, that at each step select the
same control action in all states, and thus only depend on the number
of execution steps of the system.  Several problems for MDPs with
word-strategies (also known as probabilistic automata) are
undecidable~\cite{BBG08,GO10,KVAK10,DMS11Err}. In particular the limit-sure
reachability problem, which is to decide whether a given state can be
reached with probability arbitrarily close to one, is undecidable for
probabilistic automata~\cite{GO10}.


We establish the decidability and optimal complexity of deciding
synchronizing properties for the symbolic semantics of MDPs under
general strategies.  Synchronizing properties require that the
probability distributions tend to accumulate all the probability mass
in a single state, or in a set of states. They generalize
synchronizing properties of finite automata~\cite{Volkov08,DMS11b}.
Formally for $0  \leq p \leq 1$, a sequence $\bar{X} = X_0
X_1 \dots$ of probability distributions $X_i : Q \to [0,1]$ over state
space~$Q$ of an MDP is \emph{eventually $p$-synchronizing} if for some $i \geq
0$, the distribution $X_i$ assigns probability at least~$p$ to some
state. Analogously, it is \emph{always $p$-synchronizing} if in all
distributions $X_i$, there is a state with probability at least~$p$.
For $p=1$, these definitions are the qualitative analogous for
sequences of distributions of the traditional reachability and safety
conditions~\cite{AHK07}. In particular, an eventually $1$-synchronizing
sequence witnesses that there is a length $\ell$ such that all paths 
of length $\ell$ in the MDP reach a single state, which is thus reached 
synchronously no matter the probabilistic choices.


Viewing MDPs as one-player stochastic games, we consider the following
traditional winning modes (see also Table~\ref{tab:def-modes}): $(i)$
\emph{sure} winning, if there is a strategy that generates an
\{eventually, always\} $1$-synchronizing sequence; $(ii)$
\emph{almost-sure} winning, if 
there exists a strategy that generates a
sequence that is, for all $\epsilon>0$, \{eventually, always\}
$(1-\epsilon)$-synchronizing; $(iii)$ \emph{limit-sure} winning, if
for all $\epsilon>0$, there is a strategy that generates an
\{eventually, always\} $(1-\epsilon)$-synchronizing sequence.


We show that the three winning modes form a strict hierarchy for
eventually synchronizing: there are limit-sure winning
MDPs that are not almost-sure winning, and 
there are almost-sure winning MDPs that are not sure winning. 
For always synchronizing, the three modes coincide.



For each winning mode, we consider the problem of deciding if a given
initial distribution is winning.
We establish the decidability and optimal complexity bounds for all
winning modes.  Under general strategies, the decision problems have
much lower complexity than with word-strategies. We show that all
decision problems are decidable, in polynomial time for always
synchronizing, and PSPACE-complete for eventually synchronizing.  This
is also in contrast with almost-sure winning in the traditional
semantics of MDPs, which is solvable in polynomial time for both
safety and reachability~\cite{deAlfaro97}.
All complexity results are shown in Table~\ref{tab:complex-strat}.

We complete the picture by providing optimal memory bounds for 
winning strategies. We show that for sure winning strategies, 
exponential memory is sufficient and may be necessary,
and that in general infinite memory is necessary for almost-sure winning,
and unbounded memory is necessary for limit-sure winning.

\begin{table}[!t]
\begin{center}
\scalebox{0.85}{
\begin{tabular}{|l@{\ }c@{\;} |*{6}{cc@{\;}c|}}
\hline                        
 \large{\strut} &  & \multicolumn{3}{|c|}{Always} & \multicolumn{3}{|c|}{Eventually} \\
\hline
 Sure \large{\strut}	&  
			& $\exists \alpha$ & $\forall n$                    & $\M^{\alpha}_n(T)=1$ 
        	        & $\exists \alpha$ & $\exists n$                    & $\M^{\alpha}_n(T)=1$ 
 \\
\hline
 Almost-sure \  \large{\strut}	&  
				& $\exists \alpha$ &$\inf_{n}$                 & $\M^{\alpha}_n(T)=1$
				& $\exists \alpha$& $\sup_{n}$                 & $\M^{\alpha}_n(T)=1$ 
 \\
\hline
 Limit-sure \large{\strut}   &   
		& $\sup_{\alpha}$ &$\inf_{n}$                & $\M^{\alpha}_n(T)=1$
		& $\sup_{\alpha}$& $\sup_{n}$                & $\M^{\alpha}_n(T)=1$  \\
\hline
\end{tabular}  
}
\end{center}
\caption{Winning modes and synchronizing objectives 
(where $\M^{\alpha}_n(T)$ denotes the probability that under strategy $\alpha$,
after $n$ steps the MDP $\M$ is in a state of $T$). \label{tab:def-modes}}
\end{table}

Some results in this paper rely on insights related to games and 
alternating automata that are of independent interest. 
First, the sure-winning problem for eventually synchronizing is equivalent
to a two-player game with a synchronized reachability objective, where the goal 
for the first player is to ensure that a target state is reached 
after a number of steps
%
%
that is independent of the strategy of the opponent (and thus this number can be fixed in 
advance by the first player). This condition is stronger than plain reachability,
and while the winner in two-player reachability games can be decided in polynomial time,
deciding the winner for synchronized reachability is PSPACE-complete. This result
is obtained by turning the synchronized reachability game into a one-letter alternating
automaton for which the emptiness problem (i.e., deciding if there exists a word accepted
by the automaton) is PSPACE-complete~\cite{Holzer95,AFA1}.
Second, our PSPACE lower bound for the limit-sure winning problem in eventually synchronizing
uses a PSPACE-completeness result that we establish for the \emph{universal finiteness problem},
which is to decide, given a one-letter alternating automata, whether from every 
state the accepted language is finite.


\section{Markov Decision Processes and Synchronization}\label{sec:def}

A \emph{probability distribution} over a finite set~$S$ is a
function $d : S \to [0, 1]$ such that $\sum_{s \in S} d(s)= 1$. 
The \emph{support} of~$d$ is the set $\Supp(d) = \{s \in S \mid d(s) > 0\}$. 
We denote by $\dist(S)$ the set of all probability distributions over~$S$. 
Given a set $T\subseteq S$, let $d(T)=\sum_{s\in T}d(s)$.
For $T \neq \emptyset$, the \emph{uniform distribution} on $T$ assigns probability 
$\frac{1}{\abs{T}}$ to every state in $T$.
Given $s \in S$, the \emph{Dirac distribution} on~$s$ assigns probability~$1$
to~$s$, and by a slight abuse of notation, we denote it simply by $s$.

\subsection{Markov decision processes}

A \emph{Markov decision process} (MDP) $\M = \tuple{Q,\Act,\delta}$ 
consists of a finite set $Q$ of states,  a finite set $\Act$ of actions, 
and a probabilistic transition function $\delta: Q \times \Act \to \dist(Q)$.   
A state $q$ is \emph{absorbing} if $\delta(q,a)$ is the Dirac distribution on~$q$
for all actions~$a \in \Act$.


We describe the behavior of an MDP as a one-player stochastic game
played for infinitely many rounds. 
Given an initial distribution $\mu_0 \in \dist(Q)$, the game 
starts in the first round in state $q$ with probability $\mu_{0}(q)$.
In each round, the player chooses an action $a \in \Act$,
and if the game is in state $q$, the next round starts in the 
successor state $q'$ with probability $\delta(q,a)(q')$.


Given $q\in Q$ and $a \in \Act$, denote by $\post(q,a)$ the set
$\Supp(\delta(q,a))$, and given $T \subseteq Q$ let $\Pre(T)= \{q \in
Q \mid \exists a \in \Act: \post(q, a) \subseteq T\}$ be the set of
states from which the player has an action to ensure that the
successor state is in $T$.  For $k>0$, let $\Pre^k(T) =
\Pre(\Pre^{k-1}(T))$ with $\Pre^0(T)=T$.

A \emph{path} in $\M$ is an infinite sequence $\pi = q_{0} a_{0} q_{1} a_1 \dots$
such that $q_{i+1} \in \post(q_{i},a_{i})$ 
for all $i \geq 0$. A finite prefix $\rho = q_{0} a_{0} q_{1} a_1 \dots q_n$
of a path has length $\abs{\rho}=n$ and last state $\Last(\rho)=q_{n}$.
We denote by $\Plays(\M)$ and $\Pref(\M)$ the set of all
paths and finite paths in~$\M$ respectively.

For the decision problems considered in this paper, only the support
of the probability distributions in the transition function is
relevant (i.e., the exact value of the positive probabilities does not
matter); therefore, we can encode an MDP as an $\Act$-labelled
transition system $(Q,R)$ with $R \subseteq Q \times \Act \times Q$
such that $(q,a,q') \in R$ is a transition if $q' \in \post(q,a)$.

 



\paragraph{Strategies.}
A \textit{randomized strategy} for $\M$ (or simply a strategy) is a function 
$\alpha: \Pref(\M) \to \dist(\Act)$ that, given a finite path $\rho$,
returns a probability distribution $\alpha(\rho)$ over the action set,
used to select a successor state $q'$ of $\rho$ with probability 
$\sum_{a \in \Act} \alpha(\rho)(a) \cdot \delta(q,a)(q')$ where $q=\Last(\rho)$.

A strategy $\alpha$ is
\emph{pure} if for all $\rho \in \Pref(\M)$, 
there exists an action $a \in \Act$ such that $\alpha(\rho)(a)=1$; and 
\emph{memoryless} if $\alpha(\rho) = \alpha(\rho')$
for all $\rho, \rho'$ such that $\Last(\rho) = \Last(\rho')$.
We view pure strategies as functions $\alpha: \Pref(\M)\to \Act$, and 
memoryless strategies as functions $\alpha: Q \to \dist(\Act)$,
Finally, a strategy $\alpha$ uses \emph{finite-memory} if it can be represented 
by a finite-state transducer $T = \tuple{\mem,m_0, \alpha_u, \alpha_n}$ where
$\mem$ is a finite set of modes (the memory of the strategy), 
$m_0 \in \mem$ is the initial mode,
$\alpha_u: \mem \times (\Act \times Q) \to \mem$ is an update function, 
that given the current memory, last action and state updates the memory,
and $\alpha_n: \mem \times Q \to \dist(\Act)$ is a next-move function
that selects the probability distribution $\alpha_n(m,q)$ over
actions when the current mode is $m$ and the current state of $\M$ is $q$.
For pure strategies, we assume that $\alpha_n: \mem \times Q \to \Act$.
The \emph{memory size} of the strategy is the number $\abs{\mem}$ of modes.
For a finite-memory strategy $\alpha$, let $\M(\alpha)$ be the Markov chain
obtained as the product of $\M$ with the transducer defining $\alpha$.
We assume general knowledge of the reader about Markov chains, 
such as recurrent and transient states, periodicity, and 
stationary distributions~\cite{PK11}.

\subsection{Traditional semantics}

In the traditional semantics, given an initial distribution $\mu_0 \in \dist(Q)$ 
and a strategy $\alpha$ in an MDP $\M$,
a \emph{path-outcome} is a path
$\pi = q_{0}  a_{0} q_{1} a_1 \dots$ in $\M$
such that $q_0 \in \Supp(\mu_0)$
and $a_{i} \in \Supp(\alpha(q_0 a_0 \dots q_{i}))$ 
for all $i \geq 0$. The probability of a finite prefix 
$\rho = q_{0} a_{0} q_{1} a_1 \dots q_n$ of $\pi$ is 
$$\mu_0(q_{0}) \cdot \prod_{j=0}^{n-1} \alpha(q_{0} a_{0}\dots q_{j})(a_{j}) \cdot
\delta(q_{j},a_{j})(q_{j+1}).$$
We denote by $\Outcomes(\mu_0,\alpha)$ the set of all path-outcomes from $\mu_0$ under strategy $\alpha$.
An \emph{event} $\Omega \subseteq \Plays(\M)$ is a measurable set of paths, and given an initial distribution 
$\mu_0$ and a strategy $\alpha$, the probabilities $Pr^{\alpha}(\Omega)$ of events
$\Omega$ are uniquely defined~\cite{Vardi-focs85}.
In particular, given a set $T \subseteq Q$ of target states, and $k \in \nat$, 
we denote by 
$\Box T = \{q_{0} a_{0} q_{1} \dots \in \Plays(\M) \mid \forall i: q_i \in T\}$
the safety event of always staying in $T$, by
$\Diamond T = \{q_{0} a_{0} q_{1} \dots \in \Plays(\M) \mid \exists i: q_i \in T\}$
the event of reaching $T$, and by
$\Diamond^{k}\, T = \{q_{0} a_{0} q_{1} \dots \in \Plays(\M) \mid q_k \in T \}$ the event
of reaching $T$ after exactly $k$ steps.
Hence, $\Pr^{\alpha}(\Diamond T)$ is the probability to reach $T$ under strategy $\alpha$.

We consider the following classical winning modes. 
Given an initial distribution $\mu_0$ and an event $\Omega$, we say that $\M$ is:
\begin{itemize}
\item \emph{sure winning} if there exists a strategy $\alpha$ such that
$\Outcomes(\mu_0,\alpha) \subseteq \Omega$;
\item \emph{almost-sure winning} if there exists a strategy $\alpha$ such that $\Pr^{\alpha}(\Omega) = 1$;
\item \emph{limit-sure winning} if $\sup_{\alpha} \Pr^{\alpha}(\Omega) = 1$.
\end{itemize}

It is known for safety
objectives $\Box T$ in MDPs that the three winning modes coincide, 
and for reachability objectives $\Diamond T$ that an MDP is  
almost-sure winning if and only if it is limit-sure winning. 
For both objectives, the set of initial distributions for which an MDP 
is sure (resp., almost-sure or limit-sure) winning can 
be computed in polynomial time~\cite{deAlfaro97}.

\subsection{Symbolic semantics}

In contrast to the traditional semantics, we consider a symbolic semantics
where fixing a strategy in an MDP $\M$ produces a sequence of probability 
distributions over states defined as follows~\cite{KVAK10}. 
Given an initial distribution $\mu_0 \in \dist(Q)$ and a strategy $\alpha$ in $\M$,
the \emph{symbolic outcome} of $\M$ from $\mu_0$ is 
the sequence $(\M^{\alpha}_n)_{n\in \nat}$ of probability distributions defined by 
$\M^{\alpha}_k(q) = Pr^{\alpha}(\Diamond^{k}\, \{q\})$ for all $k \geq 0$ and $q \in Q$.
Hence, $\M^{\alpha}_k$ is the probability distribution over states after $k$ steps
under strategy $\alpha$.
Note that $\M^{\alpha}_0 = \mu_0$.

Informally, synchronizing objectives require that the probability of some state 
(or some group of states) tends to $1$ in the sequence $(\M^{\alpha}_n)_{n\in \nat}$. 
Given a set $T \subseteq Q$, consider the functions $\fsum_T: \dist(Q) \to [0,1]$ 
and $\fmax_T: \dist(Q) \to [0,1]$ that compute $\fsum_T(X) = \sum_{q \in T} X(q)$
and $\fmax_T(X) = \max_{q \in T} X(q)$. 
For $f \in \{\fsum_T,\fmax_T\}$ and $p \in [0,1]$, 
we say that a probability distribution $X$ is $p$-synchronized according to $f$ if $f(X) \geq p$,
and that a sequence $\bar{X} = X_0 X_1 \dots$ of
probability distributions is:

\begin{itemize}
\item[$(a)$] \emph{always $p$-synchronizing} if $X_i$ is $p$-synchronized for all $i \geq 0$;
\item[$(b)$] \emph{eventually $p$-synchronizing} if $X_i$ is $p$-synchronized for some $i \geq 0$.
\end{itemize}
For $p=1$, we view these definitions as the qualitative analogous 
for sequences of distributions of the traditional safety and reachability
conditions for sequences of states~\cite{AHK07}. Now, 
we define the following winning modes.
Given an initial distribution $\mu_0$ and a function $f \in \{\fsum_T,\fmax_T\}$, 
we say that for the objective of \{always, eventually\} synchronizing, $\M$ is:
\begin{itemize}
\item \emph{sure winning} if there exists a strategy $\alpha$ such that
the symbolic outcome of $\alpha$ from $\mu_0$ 
is \{always, eventually\} $1$-synchronizing according to~$f$;
\item \emph{almost-sure winning} if there exists a strategy $\alpha$ such that 
for all $\epsilon>0$ the symbolic outcome of $\alpha$ from $\mu_0$ is 
\{always, eventually\} $(1-\epsilon)$-synchronizing according to $f$;
\item \emph{limit-sure winning} if for all $\epsilon>0$, there exists a strategy $\alpha$
such that the symbolic outcome of $\alpha$ from $\mu_0$ is 
\{always, eventually\} $(1-\epsilon)$-synchronizing according to $f$;
\end{itemize}

We often use $X(T)$ instead of $\fsum_T(X)$, as 
in Table~\ref{tab:def-modes} where 
the definitions of the various winning modes and synchronizing objectives
for $f=\fsum_T$ are summarized. In Section~\ref{sec:decision-problems},
we present an example to illustrate the definitions. 

\subsection{Decision problems}\label{sec:decision-problems}
For $f \in \{\fsum_T,\fmax_T\}$ and $\lambda \in \text{\{always, event(ually)\}}$, 
the \emph{winning region} $\winsure{\lambda}(f)$ is the set of initial distributions such that
$\M$ is sure winning for $\lambda$-synchronizing (we assume that $\M$ is clear from the context). 
We define analogously the winning regions $\winas{\lambda}(f)$ and $\winlim{\lambda}(f)$.
For a singleton $T = \{q\}$ we have $\fsum_{T} = \fmax_{T}$,
and we simply write $\win{\lambda}{\mu}(q)$ (where $\mu \in \text{\{sure, almost, limit\}}$).
We are interested in the algorithmic complexity of the \emph{membership problem}, 
which is to decide,
given a probability distribution $\mu_0$, whether $\mu_0 \in \win{\lambda}{\mu}(f)$.
As we show below, it is easy to establish the complexity of the 
membership problems for always synchronizing, while it is more tricky 
for eventually synchronizing. The complexity results are summarized
in Table~\ref{tab:complex-strat}.

\begin{table}[t]
\begin{center}
\scalebox{0.90}{
\begin{tabular}{|l@{\ }c@{\;} |*{2}{c|c|}}
\hline                        
\large{\strut}           &  & \multicolumn{2}{|c|}{Always} & \multicolumn{2}{|c|}{Eventually} \\
\cline{3-6}
\large{\strut}	&
						&								\;Complexity\;	& \;Memory requirement\;        &       \;Complexity\;	& \;Memory requirement\;  \\
\hline
 Sure \large{\strut}       			&
		&  													& 
                &           PSPACE-C    &  exponential
	 \\
\cline{1-2} \cline{5-6}
 Almost-sure \large{\strut}  				& 
		&          							PTIME		& memoryless
		& 						          PSPACE-C     & infinite
 \\
\cline{1-2} \cline{5-6}
 Limit-sure \large{\strut}   & 
		& 						              &  
		& 					            PSPACE-C    &  unbounded \\
\hline
\end{tabular}  
}
\end{center}
\caption{Computational complexity of the membership problem, 
and memory requirement for the strategies (for always synchronizing, the three 
modes coincide).\label{tab:complex-strat}}
\end{table}

\paragraph{Always synchronizing.}
We first remark that for always synchronizing, the three winning modes coincide.


\begin{lemma}\label{lem:always}
Let $T$ be a set of states. For all functions $f \in \{\fmax_T, \fsum_T\}$,  
we have
$\winsure{always}(f) = \winas{always}(f) = \winlim{always}(f)$.
\end{lemma}

\begin{proof}
It follows from the definition of winning modes that
$\winsure{always}(f) \subseteq \winas{always}(f) \subseteq \winlim{always}(f)$.
Hence it suffices to show that $\winlim{always}(f) \subseteq \winsure{always}(f)$,
that is for all $\mu_0$, if $\M$ is limit-sure always synchronizing from $\mu_0$, then 
$\M$ is sure always synchronizing from $\mu_0$. 
For $f=\fmax_T$, consider $\epsilon$ smaller than the
smallest positive probability in the initial distribution $\mu_0$ and in the 
transitions of the MDP $\M = \tuple{Q,\Act,\delta}$.
Then, given an always $(1-\epsilon)$-synchronizing 
strategy, it is easy to show by induction on $k$ 
that the distributions $\M^{\alpha}_k$ are Dirac for all $k \geq 0$. 
In particular $\mu_0$ is Dirac, and let $q_0 \in T$ be such that $\mu_0(q_0)=1$.
It follows that there is an infinite path from $q_0$
in the graph $\tuple{T,E}$ where $(q,q') \in E$ if there exists an action $a \in \Act$
such that $\delta(q,a)(q') = 1$. The existence of this
path entails that there is a loop reachable from 
$q_0$ in the graph $\tuple{T,E}$, and this naturally defines a sure-winning 
always synchronizing strategy in $\M$.
A similar argument for $f=\fsum_T$ shows that for sufficiently  small 
$\epsilon$, an always $(1-\epsilon)$-synchronizing strategy~$\alpha$ must 
produce a sequence of distributions with support contained in $T$, until
some support repeats in the sequence. This naturally induces 
an always $1$-synchronizing strategy. 
\qed
\end{proof}

It follows from the proof of Lemma~\ref{lem:always} that the winning region for  
always synchronizing according to $\fsum_T$ coincides with the set of winning initial 
distributions for the safety objective $\Box T$ in the traditional semantics,
which can be computed in polynomial time~\cite{CH12}. 
Moreover, always synchronizing according to $\fmax_T$ 
is equivalent to the existence of an infinite path staying in $T$
in the transition system $\tuple{Q,R}$ of the MDP restricted 
to transitions $(q,a,q') \in R$ such that $\delta(q,a)(q')=1$,
which can also be decided in polynomial time. 
In both cases, pure memoryless strategies are sufficient.


\begin{theorem}\label{theo:always}
The membership problem for always synchronizing can be solved in polynomial
time, and pure memoryless strategies are sufficient.
\end{theorem}

\paragraph{Eventually synchronizing.}
For all functions $f \in \{\fmax_T, \fsum_T\}$, the following inclusions
hold: $\winsure{event}(f) \subseteq \winas{event}(f) \subseteq \winlim{event}(f)$
and we show that the inclusions are strict in general.

\begin{figure}[t]
\begin{center}
    \begin{picture}(80,15)(0,0)

\node[Nmarks=i,iangle=180](n0)(10,2){$q_0$}
\node[Nmarks=n](n1)(32,2){$q_1$}
\node[Nmarks=n](n2)(52,2){$q_2$}
\node[Nmarks=n](n3)(72,2){$q_3$}

\drawedge(n0,n1){$a,b:\!\sfrac{1}{2}$}
\drawloop[ELside=l,loopCW=y, loopangle=90, loopdiam=4](n0){$a,b:\!\sfrac{1}{2}$}

\drawedge(n1,n2){$b$}
\drawloop[ELside=l,loopCW=y, loopangle=90, loopdiam=4](n1){$a$}

\drawedge(n2,n3){$a,b$}
\drawloop[ELside=l,loopCW=y, loopangle=90, loopdiam=4](n3){$a,b$}

\end{picture}
\end{center}
 \caption{An MDP~$\M$  such that $\winsure{event}(q_1) \neq \winas{event}(q_1)$ 
and $\winas{event}(q_2) \neq \winlim{event}(q_2)$. \label{fig:almost-limit-eventually-differ}}
\end{figure}

\begin{lemma}\label{lem:dif-in-def}
There exists an MDP $\M$ and states $q_1,q_2$ such that
$(i)$ $\winsure{event}(q_1) \subsetneq \winas{event}(q_1)$, and 
$(ii)$ $\winas{event}(q_2) \subsetneq \winlim{event}(q_2)$. 
\end{lemma}

\begin{proof}
Consider the MDP $\M$ with states $q_0, q_1, q_2, q_3$ and actions~$a,b$ as shown in 
\figurename~\ref{fig:almost-limit-eventually-differ}. 
All transitions are deterministic except from $q_0$ where on all actions, 
the successor is $q_0$ or $q_1$ with probability $\frac{1}{2}$. 
Let   the initial distribution~$\mu_0$ be  a Dirac distribution on~$q_0$.

To establish $(i)$, we show  that 
$\mu_0 \in \winas{event}(q_1)$ and $\mu_0 \not \in \winsure{event}(q_1)$.
To prove that $\mu_0 \in \winas{event}(q_1)$, consider the pure 
strategy that always plays~$a$. The outcome is such that the probability
to be in~$q_1$ after~$k$ steps is $1-\frac{1}{2^k}$, showing that $\M$
is almost-sure winning for the eventually synchronizing objective in~$q_1$ (from~$\mu_0$).
On the other hand, $\mu_0 \not \in \winsure{event}(q_1)$ because
for all strategies~$\alpha$, the probability in~$q_0$ remains always positive,
and thus in~$q_1$ we have $\M^{\alpha}_n(q_1) < 1$ for all $n \geq 0$,
showing that~$\M$ is not sure winning for the eventually synchronizing 
objective in~$q_1$ (from~$\mu_0$).

To establish $(ii)$, for all $k \geq 0$ consider a strategy that plays $a$ 
for $k$ steps, and then plays $b$. Then the probability
to be in $q_2$ after $k+1$ steps is $1-\frac{1}{2^k}$, showing that this strategy
is eventually $(1-\frac{1}{2^k})$-synchronizing in $q_2$. Hence, $\M$ is limit-sure winning 
for the eventually synchronizing objective in $q_2$ (from $\mu_0$).
Second, for all strategies, since the probability in $q_0$ remains always 
positive, the probability in $q_2$ is always smaller than $1$. Moreover,
if the probability $p$ in $q_2$ is positive after $n$ steps ($p>0$), 
then after any number $m > n$ of steps, the probability in $q_2$ is bounded by
$1-p$. It follows that the probability in $q_2$ is never equal to $1$ and 
cannot tend to $1$ for $m \to \infty$, showing that $\M$ is not almost-sure winning for 
the eventually synchronizing objective in $q_2$ (from $\mu_0$).
\qed
\end{proof}

The rest of this paper is devoted to the solution of the membership
problem for eventually synchronizing.  We make some preliminary
remarks to show that it is sufficient to solve the membership problem
according to $f = \fsum_T$ and for MDPs with a single initial state.
Our results will also show that pure strategies are sufficient in
all modes.

\paragraph{Remark.}
For eventually synchronizing and each winning mode, 
we show that the membership problem with function $\fmax_T$
is polynomial-time equivalent to the membership problem 
with function $\fsum_{T'}$ with a singleton ${T'}$.
First, for $\mu \in \text{\{sure, almost, limit\}}$,
we have $\win{event}{\mu}(\fmax_T) = \bigcup_{q \in T} \win{event}{\mu}(q)$,
showing that the membership problems for $\fmax$ 
are polynomial-time reducible to the corresponding membership problem 
for $\fsum_T$ with singleton~$T$. 
The reverse reduction is as follows. 
Given an MDP $\M$, a state $q$ and an initial distribution $\mu_0$,
we can construct an MDP $\M'$ and initial distribution $\mu'_0$
such that $\mu_0 \in \win{event}{\mu}(q)$ iff $\mu'_0 \in \win{event}{\mu}(\fmax_{Q'})$
where $Q'$ is the state space of $\M'$. The idea is to construct $\M'$ and $\mu'_0$
as a copy of~$\M$ and~$\mu_0$ where all states except $q$ are duplicated, and the 
initial and transition probabilities are equally distributed between 
the copies (see \figurename~\ref{fig:twin}). 
Therefore if the probability tends to~$1$ in some state, it has to be in~$q$.

\paragraph{Remark.}
To solve the membership problems for eventually synchronizing 
with function $\fsum_T$, it is sufficient to provide 
an algorithm that decides membership of Dirac distributions (i.e., assuming
MDPs have a single initial state), since  to solve the problem for an MDP $\M$ 
with initial distribution $\mu_0$, we can equivalently solve it for a copy of $\M$ with a
new initial state $q_{0}$ from which the successor distribution 
on all actions is~$\mu_0$. 
%
Therefore, it is sufficient to consider initial Dirac distributions $\mu_0$. 

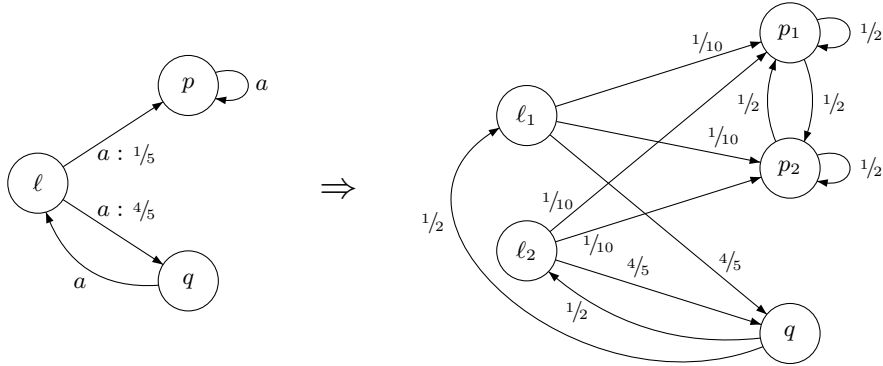
\begin{figure}[t]
\begin{center}
    \begin{picture}(120,50)(0,2)

\node[Nmarks=n](l)(5,25){$\ell$}
\node[Nmarks=n](p)(25,38){$p$}
\node[Nmarks=n](q)(25,12){$q$}

\drawedge[ELside=r,ELdist=0](l,p){$a:\sfrac{1}{5}$}
\drawedge[ELside=l,ELdist=0](l,q){$a:\sfrac{4}{5}$}

\drawedge[ELside=l, curvedepth=6](q,l){$a$}
\drawloop[ELside=l,loopCW=y, loopangle=0, loopdiam=4](p){$a$}

\node[Nframe=n](arrow)(45,24){{\Large $\Rightarrow$}}

\node[Nmarks=n](l1)(70,34){$\ell_1$}
\node[Nmarks=n](l2)(70,16){$\ell_2$}
\node[Nmarks=n](p1)(105,45){$p_1$}
\node[Nmarks=n](p2)(105,27){$p_2$}
\node[Nmarks=n](q)(105,5){$q$}

\drawedge[ELside=l,ELpos=70, ELdist=0](l1,p1){$\sfrac{1}{10}$}
\drawedge[ELside=l,ELpos=73, ELdist=.5](l1,p2){$\sfrac{1}{10}$}
\drawedge[ELside=l,ELpos=73, ELdist=0](l1,q){$\sfrac{4}{5}$}

\drawedge[ELside=l,ELpos=16, ELdist=0](l2,p1){$\sfrac{1}{10}$}
\drawedge[ELside=r,ELpos=25, ELdist=0](l2,p2){$\sfrac{1}{10}$}
\drawedge[ELside=l,ELpos=40, ELdist=.5](l2,q){$\sfrac{4}{5}$}

\drawbpedge[ELpos=70,ELside=l](q,210,30,l1,200,30){$\sfrac{1}{2}$} 
\drawedge[ELside=l, ELpos=75, ELdist=0, curvedepth=5](q,l2){$\sfrac{1}{2}$}

\drawloop[ELside=l,loopCW=y, loopangle=0, loopdiam=4](p1){$\sfrac{1}{2}$}
\drawedge[ELside=l,ELpos=52, curvedepth=3](p1,p2){$\sfrac{1}{2}$}
\drawedge[ELside=l,ELpos=48, curvedepth=3](p2,p1){$\sfrac{1}{2}$}
\drawloop[ELside=l,loopCW=y, loopangle=0, loopdiam=4](p2){$\sfrac{1}{2}$}

\end{picture}
\end{center}
 \caption{
State duplication ensures that the  probability mass can never 
be accumulated in a single state except in~$q$ (we omit action $a$ for readability).
\label{fig:twin}}
\end{figure}

\section{One-Letter Alternating Automata}\label{sec:1L-AFA}

In this section, we consider \emph{one-letter alternating automata} (1L-AFA)
as they have a structure of alternating graph analogous to MDP (i.e., when
ignoring the probabilities). We review classical decision problems
for 1L-AFA, and establish the complexity of a new problem, the 
\emph{universal finiteness problem} which is to decide if from every initial
state the language of a given 1L-AFA is finite. 
These results of independent interest are useful to establish the PSPACE lower 
bounds for eventually synchronizing in MDPs.

\paragraph{One-letter alternating automata.}
Let $\Bool(Q)$ be the set of positive Boolean formulas over $Q$, i.e. Boolean
formulas built from elements in $Q$ using $\land$ and $\lor$. A set 
$S \subseteq Q$ \emph{satisfies} a formula $\varphi \in \Bool(Q)$ (denoted $S \models \varphi$)
if $\varphi$ is satisfied when replacing in $\varphi$ the elements in $S$ by \true, 
and the elements in $Q \setminus S$ by \false.

A \emph{one-letter alternating finite automaton}  is a tuple 
$\A=\tuple{Q,\delta_{\A},\F}$ where 
$Q$ is a finite set of states, 
$\delta_{\A}: Q \to \Bool(Q)$ is the transition function, 
and $\F \subseteq Q$ is the set of accepting states. 
We assume that the formulas in transition function are in disjunctive normal form.
Note that the alphabet of the automaton is omitted, as it has a single letter. 
In the language of a 1L-AFA, only the length of words is relevant. 
For all $n\geq 0$, define the set $Acc_{\A}(n,\F) \subseteq Q$ of states from which 
the word of length $n$ is accepted by $\A$ as follows:
\begin{itemize}
	\item $Acc_{\A}(0,\F) = \F$;
	\item $Acc_{\A}(n,\F) = \{q \in Q \mid Acc_{\A}(n-1,\F) \models \delta(q) \}$ for all $n > 0$.
\end{itemize}

\noindent The set $\LL(\A_q) = \{n \in \nat \mid q \in Acc_{\A}(n,\F)\}$ is the 
\emph{language} accepted by $\A$ from initial state~$q$.
%

For fixed $n$, we view $Acc_{\A}(n,\cdot)$ as an operator on $2^Q$
that, given a set $\F \subseteq Q$ computes the set $Acc_{\A}(n,\F)$. 
Note that $Acc_{\A}(n,\F) = Acc_{\A}(1,Acc_{\A}(n-1,\F))$ for all $n \geq 1$.
Denote by $\Pre_{\A}(\cdot)$ the operator $Acc_{\A}(1,\cdot)$. 
Then for all $n \geq 0$ the operator $Acc_{\A}(n,\cdot)$ coincides
with $\Pre^n_{\A}(\cdot)$, the $n$-th iterate of $\Pre_{\A}(\cdot)$.


\paragraph{Decision problems.}
We present classical decision problems for alternating automata, namely 
the emptiness and finiteness problems, and we introduce a variant of
the finiteness problem that will be useful for solving synchronizing 
problems for MDPs.

\begin{itemize}
\item 
The \emph{emptiness problem} for 1L-AFA is to decide, given a 1L-AFA $\A$
and an initial state $q$, whether $\LL(\A_q)=\emptyset$. 
The emptiness problem can be solved by checking whether $q \in \Pre^n_{\A}(\F)$
for some $n \geq 0$. 
It is known that the emptiness problem is PSPACE-complete,
even for transition functions in disjunctive normal form~\cite{Holzer95,AFA1}.  

\item 
The \emph{finiteness problem} is to decide, given a 1L-AFA $\A$
and an initial state $q$, whether $\LL(\A_q)$ is finite.
The finiteness problem can be solved in (N)PSPACE by 
guessing $n,k \leq 2^{\abs{Q}}$ such that $\Pre^{n+k}_{\A}(\F) = \Pre^n_{\A}(\F)$
and $q \in \Pre^n_{\A}(\F)$. The finiteness problem is PSPACE-complete
by a simple reduction from the emptiness problem: from an instance $(\A,q)$
of the emptiness problem, construct $(\A',q')$ where $q'=q$ and 
$\A' = \tuple{Q,\delta',\F}$ is a copy of $\A = \tuple{Q,\delta,\F}$ 
with a self-loop on~$q$ (formally, $\delta'(q) = q \lor \delta(q)$
and $\delta'(r) = \delta(r)$ for all $r \in Q \setminus \{q\}$).
It is easy to see that $\LL(\A_q)=\emptyset$ iff $\LL(\A'_{q'})$ is finite.

\item 
The \emph{universal finiteness problem} 
is to decide, given a 1L-AFA $\A$, whether $\LL(\A_q)$ is finite
for all states $q$. This problem can be solved by checking whether 
$\Pre^n_{\A}(\F) = \emptyset$ for some $n \leq 2^{\abs{Q}}$, 
and thus it is in PSPACE.
Note that if $\Pre^n_{\A}(\F) = \emptyset$, then $\Pre^m_{\A}(\F) = \emptyset$
for all $m \geq n$.
\end{itemize}

Given the PSPACE-hardness proofs of the emptiness and finiteness
problems, it is not easy to see that the universal finiteness  
problem is PSPACE-hard.

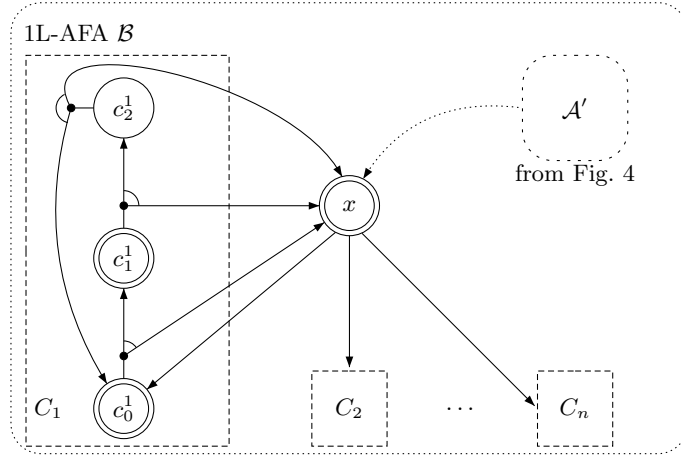
\begin{figure}[t]
\begin{center}
    \begin{picture}(90,60)(0,2)

\node[Nmarks=n, Nw=90, Nh=60, dash={0.2 0.5}0](m1)(45,29){}
\node[Nframe=n](label)(9,55){1L-AFA $\B$}

\node[Nmarks=r](n1)(45,32){$x$}

\drawpolygon[dash={0.8 0.5}0](40,10)(40,0)(50,0)(50,10)
\node[Nframe=n,Nw=1, Nh=1](n3)(45,10){}
\node[Nframe=n,Nw=1, Nh=1](nm3)(47,10){}
\node[Nframe=n](label)(45,5){$C_2$}

\node[Nframe=n](label)(60,5){$\dots $}

\drawpolygon[dash={0.8 0.5}0](70,10)(70,0)(80,0)(80,10)
\node[Nframe=n,Nw=1, Nh=1](n4)(70,5){}
\node[Nframe=n,Nw=1, Nh=1](nm4)(75,10){}
\node[Nframe=n](label)(75,5){$C_n$}

\node[Nmarks=n, Nw=14, Nh=14, dash={0.4 1}0](nm1)(75,45){}
\node[Nframe=n,Nw=1, Nh=1](aaa)(68,45){}
\node[Nframe=n](label)(75,45){$\A'$}
\node[Nframe=n](label)(75,36){from \figurename~\ref{fig:pre-empty-reduction}}

\drawedge(n1,n3){}
\drawedge[syo=-2](n1,n4){}

\drawpolygon[dash={0.8 0.5}0](2,52)(29,52)(29,0)(2,0)
\node[Nframe=n](n2)(5,5){$C_1$}

\node[Nmarks=r](c0)(15,5){$c_0^1$}
\node[Nw=1, Nh=1,Nfill=y](cc0)(15,12){}
\node[Nmarks=r](c1)(15,25){$c_1^1$}
\node[Nw=1, Nh=1,Nfill=y](cc1)(15,32){}
\node[Nmarks=n](c2)(15,45){$c_2^1$}
\node[Nw=1, Nh=1,Nfill=y](cc2)(8,45){}
\drawedge[AHnb=0](c0,cc0){}
\drawedge[AHnb=0](c1,cc1){}
\drawedge[AHnb=0](c2,cc2){}

\drawedge[ELpos=50, ELside=r, syo=-2](n1,c0){}
\drawedge[ELpos=50, ELside=r](cc0,n1){}
\drawedge[ELside=r](cc0,c1){}
\drawarc[](15,12,2,33.7,90)

\drawedge[ELpos=50, ELside=r](cc1,n1){}
\drawedge[ELside=r](cc1,c2){}
\drawarc[](15,32,2,0,90)

\drawbpedge[ELpos=70,ELside=r](cc2,120,15,n1,90,15){} 
\drawedge[ELpos=50, ELside=r,   curvedepth=-5](cc2,c0){}
\drawarc[](8,45,2,116,257)

\drawedge[ELpos=50, ELside=r, dash={0.2 0.5}0,  curvedepth=-5](aaa,n1){}

\end{picture}
\end{center}
 \caption{Sketch of reduction to show PSPACE-hardness of the universal finiteness problem for
1l-AFA.}
\label{fig:pre-empty-reduction2} 
\end{figure}

\begin{lemma}\label{lem:universal-finiteness-pspace-hard}
The universal finiteness problem for 1L-AFA is PSPACE-hard.
\end{lemma}

\begin{proof}
The proof is by a reduction from the emptiness problem for 1L-AFA,
which is PSPACE-complete~\cite{Holzer95,AFA1}. The language of a
1L-AFA~$\A = \tuple{Q, \delta, \F}$ is non-empty if $q_0 \in
\Pre_{\A}^i(\F)$ for some $i\geq 0$.  Since the sequence
$\Pre_{\A}^i(\F)$ is ultimately periodic, it is sufficient to compute
$\Pre_{\A}^i(\F)$ for all $i\leq 2^{\abs{Q}}$.

From $\A$, we construct a 1L-AFA~$B = \tuple{Q', \delta', \F'}$ with
set $\F'$ of accepting states such that the sequence $\Pre_B^i(\F')$
in $B$ mimics the sequence $\Pre_{\A}^i(\F)$ in $\A$ for $2^{\abs{Q}}$
steps.  The automaton~$B$ contains the state space of $\A$, i.e. $Q
\subseteq Q'$.  The goal is to have $\Pre_B^i(\F') \cap Q =
\Pre_{\A}^i(\F)$ for all $i \leq 2^{\abs{Q}}$, as long as $q_0 \not\in
\Pre_{\A}^i(\F)$. Moreover, if $q_0 \in \Pre_{\A}^i(\F)$ for some
$i\geq 0$, then $\Pre_B^j(\F')$ will contain $q_0$ for all $j\geq i$
(the state $q_0$ has a self-loop in $B$), and if $q_0 \not\in
\Pre{\A}^i(\F)$ for all $i\geq 0$, then $B$ is constructed such that
$\Pre_B^j(\F') = \emptyset$ for sufficiently large~$j$ (roughly for $j
> 2^{\abs{Q}}$).  Hence, the language of $\A$ is non-empty if and only
if the sequence $\Pre_B^j(\F')$ is not ultimately empty, that is if
and only if the language of $B$ is infinite from some state (namely
$q_0$).

The key is to let $B$ simulate $\A$ for exponentially many steps, and 
to ensure that the simulation stops if and only if $q_0$ is not reached within $2^{\abs{Q}}$ steps.
We achieve this by defining $B$ as the gadget in~\figurename~\ref{fig:pre-empty-reduction2}
connected to a modified copy $\A'$ of $\A$ with the same state space. 
The transitions in $\A'$ are defined as follows, where $x$ is the entry state
of the gadget (see \figurename~\ref{fig:pre-empty-reduction}):
for all $q \in Q$ let $(i)$ $\delta_{B}(q) = x \land \delta_{\A}(q)$ if $q \neq q_0$,
and $(ii)$ $\delta_{B}(q_0) = q_0 \lor (x \land \delta_{\A}(q_0))$.
%
Thus, $q_0$ has a self-loop, and given a set $S \subseteq Q$ in the automaton $\A$, 
if $q_0 \not\in S$, then $\Pre_{\A}(S) = \Pre_B(S \cup \{x\})$ that is 
$\Pre_B$ mimics $\Pre_{\A}$ when $x$ is in the argument (and $q_0$ has not been reached yet).
Note that if $x \not\in S$ (and $q_0 \not\in S$), 
then $\Pre_B(S) = \emptyset$, that is unless $q_0$ has been reached, the
simulation of $\A$ by $B$ stops. 
Since we need that $B$ mimics $\A$ for $2^{\abs{Q}}$ steps, we define 
the gadget and the set $\F'$ to ensure that $x \in \F'$ and 
if $x \in \Pre_B^i(\F')$, then $x \in \Pre_B^{i+1}(\F')$ for all $i \leq 2^{\abs{Q}}$. 

In the gadget, the state $x$ has nondeterministic 
transitions $\delta_{B}(x) = c^1_0 \lor c^2_0 \lor \dots \lor c^n_0$
to $n$ components with state space 
$C_i = \{c^i_0, \dots, c^i_{p_i-1} \}$ where
$p_i$ is the $(i+1)$-th prime number,
and the transitions\footnote{In expression $c^i_j$,
we assume that $j$ is interpreted modulo $p_i$.} 
$\delta_{B}(c^i_j) = x \land c^i_{j+1}$ form a loop in each component
($i=1,\dots,n$).
We choose $n$ such that $p^{\#}_n = \prod_{i=1}^{n} p_i > 2^{\abs{Q}}$ (take $n=\abs{Q}$). 
Note that the number of states in the gadget is $1 + \sum_{i=1}^{n} p_i \in O(n^2 \log n)$~\cite{BS96}
and hence the construction is polynomial in the size of $\A$. 

By construction, for all sets $S$, we have $x \in \Pre_B(S)$ whenever the first state $c^i_0$
of some component $C_i$ is in $S$, and if $x \in S$, then  
$c^i_j \in S$ implies $c^i_{j-1} \in \Pre_B(S)$.
Thus, if $x \in S$, the operator $\Pre_B(S)$ 
`shifts' backward the states in each component;
and, $x$ is in the next iteration (i.e., $x \in Pre_B(S)$) 
as long as $c^i_0 \in S$ for some component $C_i$.

Now, define the set of accepting states $\F'$ in $B$ in such a way that 
all states $c^i_0$ disappear simultaneously only after $p^{\#}_n$ iterations.
Let $\F' = \F \cup \{x\} \cup \bigcup_{1 \leq i \leq n} (C_{i} \setminus \{c^i_{p_i-1}\})$, 
thus $\F'$ contains all states of the gadget except the last state of each
component. 
It is easy to check that, irrespective of the transition relation in $\A$,
we have $x \in \Pre_B^{i}(\F')$ if and only if $0 \leq i < p_n^{\#}$. 
Therefore, 
if $q_0 \in \Pre_{\A}^{i}(\F)$ for some $i$, then $q_0 \in \Pre_B^{j}(\F')$ 
for all $j \geq i$ by the self-loop on $q_0$. On the other hand, 
if $q_0 \not\in \Pre_{\A}^{i}(\F)$ for all $i \geq 0$, then since $x \not\in 
\Pre_B^{i}(\F')$ for all $i > p_n^{\#}$, we have $\Pre_B^{i}(\F') = \emptyset$ 
for all $i > p_n^{\#}$. This shows that the language of $\A$ is non-empty 
if and only if the language of $B$ is infinite from some state (namely $q_0$), 
and establishes the correctness of the reduction. \qed
\end{proof}

\begin{figure}[t]
\begin{center}
    \begin{picture}(122,45)(0,2)

\node[Nmarks=n, Nw=30, Nh=42, dash={0.2 0.5}0](m1)(15,20){}
\node[Nframe=n](label)(15,2){1L-AFA $\A$}
\node[Nmarks=n](n0)(15,36){$q_0$}
\node[Nframe=n](label)(26,31){$\delta_{A}(q_0) =~q_1$}
\node[Nmarks=n](n1)(15,21){$q_1$}
\node[Nframe=n](label)(29,16){$\delta_{A}(q_1) =~q_2 \land q_3$}
\node[Nw=1, Nh=1,Nfill=y](m1)(15,15){}
\node[Nmarks=r](n2)(5,8){$q_2$}
\node[Nmarks=n](n3)(25,8){$q_3$}
\drawedge(n0,n1){}
\drawedge[AHnb=0](n1,m1){}
\drawedge[ELpos=50](m1,n2){}
\drawedge[ELpos=50](m1,n3){}
\drawarc[](15,15,2,215,325)

\node[Nframe=n](arrow)(43,23){{\Large $\Rightarrow$}}

\node[Nmarks=n, Nw=30, Nh=42, dash={0.2 0.5}0](nm1)(81,20){}
\node[Nmarks=n, Nw=67, Nh=48, dash={0.4 1}0](m2)(88,22){}
\node[Nframe=n](label)(110,43){1L-AFA $\A'$}
\node[Nframe=n](label)(81,2){1L-AFA $\A$}
\node[Nmarks=n](nn0)(81,36){$q_0$}
\node[Nframe=n](label)(101.7,31){$\delta_{B}(q_0) = q_0 \lor (x \land \delta_{\A}(q_0))$}
\node[Nw=1, Nh=1,Nfill=y](mm0)(81,30){}
\node[Nmarks=n](nn1)(81,21){$q_1$}
\node[Nframe=n](label)(97.3,16){$\delta_{B}(q_1) = x \land \delta_{\A}(q_1)$}
\node[Nw=1, Nh=1,Nfill=y](mm1)(81,15){}
\node[Nmarks=r](nn2)(71,8){$q_2$}
\node[Nmarks=n](nn3)(91,8){$q_3$}
\drawedge[AHnb=0](nn0,mm0){}
\drawedge[AHnb=0](nn1,mm1){}

\node[Nmarks=r](nn4)(60,20){$x$}

\drawedge(mm0,nn1){}
\drawedge[ELpos=50, ELside=r](mm0,nn4){}
\drawarc[](81,30,2,205.5,270)

\drawedge[ELpos=50](mm1,nn2){}
\drawedge[ELpos=50](mm1,nn3){}
\drawedge[ELpos=50, ELside=r](mm1,nn4){}
\drawarc[](81,15,2,166.6,325)
\drawloop[ELside=l,loopCW=y, loopangle=90, loopdiam=4](nn0){}

\end{picture}
\end{center}
 \caption{Detail of the copy $\A'$ obtained from $\A$ in the reduction
of \figurename~\ref{fig:pre-empty-reduction2}.
}\label{fig:pre-empty-reduction}
\end{figure}
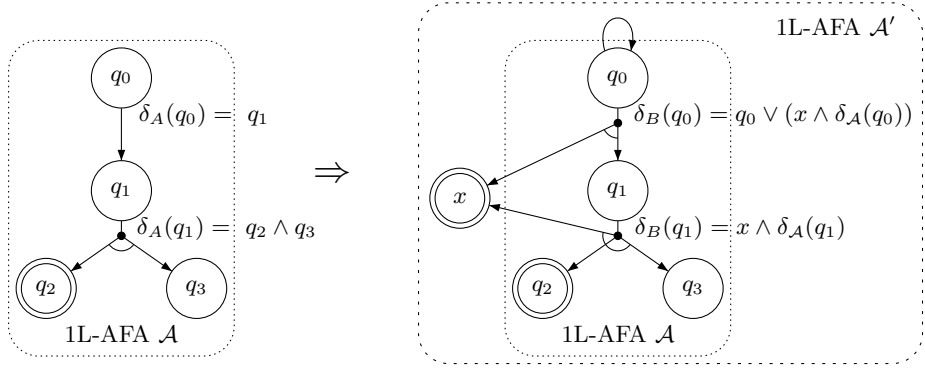

\paragraph{Relation with MDPs.}
The underlying structure of a Markov decision process $\M = \tuple{Q,\Act,\delta}$ 
is an alternating graph, where the successor $q'$ of a state $q$ is obtained by an existential 
choice of an action $a$ and a universal choice of a state $q' \in \Supp(\delta(q,a))$.
Therefore, it is natural that some questions related to MDPs 
have a corresponding formulation in terms of alternating automata. We show that
such connections exist between synchronizing problems for MDPs and language-theoretic
questions for alternating automata, such as emptiness and universal finiteness.
%
Given a 1L-AFA $\A = \tuple{Q, \delta_{\A}, \F}$, assume without loss of generality
that the transition function $\delta_{\A}$ is such that 
$\delta_{\A}(q) = c_1 \lor \dots \lor c_m$ has the same number $m$ of conjunctive clauses 
for all $q \in Q$. From $\A$, construct
the MDP $\M_{\A} = \tuple{Q, \Act, \delta_{\M}}$ where $\Act = \{a_1, \dots, a_m\}$
and $\delta_{\M}(q,a_k)$ is the uniform distribution over the states occurring
in the $k$-th clause $c_k$ in $\delta_{\A}(q)$, for all $q \in Q$ and $a_k \in \Act$.
Then, we have $Acc_{\A}(n,\F) = \Pre_{\M}^n(\F)$ for all $n \geq 0$.
Similarly, from an MDP $\M$ and a set $T$ of states, we can construct
a 1L-AFA $\A = \tuple{Q, \delta_{\A}, \F}$ with $\F = T$ such that 
$Acc_{\A}(n,\F) = \Pre_{\M}^n(T)$ for all $n \geq 0$ (let $\delta_{\A}(q) = \bigvee_{a \in \Act}
\bigwedge_{q' \in \post(q,a)} q'$ for all $q \in Q$). 

Several decision problems for 1L-AFA can be solved by computing 
the sequence $Acc_{\A}(n,\F)$, and we show that some synchronizing 
problems for MDPs require the computation of the sequence $\Pre_{\M}^n(\F)$. 
Therefore, the above relation between 1L-AFA and MDPs establishes bridges that 
we use in Section~\ref{sec:eventually} to transfer complexity results from 1L-AFA to MDPs. 


\section{Eventually Synchronization}\label{sec:eventually}

In this section, we show the PSPACE-completeness of the
membership problem for eventually synchronizing objectives
and the three winning modes. By the remarks at the end
of Section~\ref{sec:def}, we consider the membership problem
with function $\fsum$ and Dirac initial distributions (i.e., single
initial state).

\subsection{Sure  eventually synchronization}\label{sec:sure-event-sync}
Given a target set $T$, the membership problem for sure-winning eventually 
synchronizing objective in $T$ can be solved by computing the sequence 
$\Pre^n(T)$ of iterated predecessor. A state $q_0$ is sure-winning for 
eventually synchronizing in~$T$ if $q_0 \in \Pre^n(T)$ for some $n \geq 0$.

\begin{lemma}\label{lem:sure-ss-pre}
Let $\M$ be an MDP and $T$ be a target set. 
For all states $q_0$, we have $q_0 \in \winsure{event}(\fsum_T)$ 
if and only if there exists $n \geq 0$ such that $q_0 \in \Pre_{\M}^{n}(T)$. 
\end{lemma}

\begin{proof}
We prove the following equivalence by induction (on the length $i$):
for all initial states $q_0$, there exists a strategy $\alpha$ sure-winning
in $i$ steps from $q_0$ (i.e., such that $\M^{\alpha}_i(T) = 1$)
if and only if $q_0 \in \Pre^{i}(T)$. The case $i=0$ trivially holds 
since for all strategies $\alpha$, we have $\M^{\alpha}_0(T)=1$ if and only if $q_0 \in T$.

Assume that the equivalence holds for all $i < n$. 
For the induction step, show that $\M$ is sure eventually synchronizing from $q_0$ (in $n$ steps) 
if and only if there exists an action $a$ such that $\M$ is sure eventually 
synchronizing (in $n-1$ steps) from all states $q' \in \post(q_0, a)$ (equivalently,
$\post(q_0, a) \subseteq \Pre^{n-1}(T)$ by the induction hypothesis, that is 
$q_0 \in \Pre^{n}(T)$). First, if all successors $q'$ of $q_0$ under some action $a$ 
are sure eventually synchronizing, then so is $q_0$ by playing~$a$ followed by
a winning strategy from each successor $q'$. 
For the other direction, assume towards contradiction that $\M$ is sure eventually 
synchronizing from $q_0$ (in $n$ steps), but for each action~$a$, 
there is a state $q' \in \post(q_0, a)$ that is not sure eventually synchronizing. 
Then, from $q'$ there is a positive probability to reach a state 
not in $T$ after $n-1$ steps, no matter the strategy played. 
Hence from~$q_0$, for all strategies, the probability mass in $T$ cannot be $1$ 
after $n$ steps, in contradiction with the fact that $\M$ is sure eventually 
synchronizing from $q_0$ in $n$ steps.  
It follows that the induction step holds, and the proof is complete.
\qed
\end{proof}

By Lemma~\ref{lem:sure-ss-pre}, 
the membership problem for sure eventually synchronizing is
equivalent to the emptiness problem of 1L-AFA, and thus PSPACE-complete.
Moreover if $q_0 \in \Pre_{\M}^{n}(T)$, a finite-memory strategy with~$n$ modes
that at mode~$i$ in a state~$q$ plays an action~$a$ such that
$\post(q,a)\subseteq \Pre^{i-1}(T)$ is sure winning for eventually synchronizing.  

There exists a family of MDPs $\M_n$ ($n \in \nat$) over alphabet $\{a,b\}$
that are sure winning for eventually synchronization, and where the sure winning strategies 
require exponential memory. The MDP $\M_2$ is shown in \figurename~\ref{fig:exp-mem}.
The structure of $\M_n$ is an initial uniform probabilistic transition
to $n$ components $H_1, \dots, H_n$ where $H_i$ is a cycle of length $p_i$
the $i$th prime number. On action $a$, the next state in the cycle is reached,
and on action $b$ the target state $q_T$ is reached, only from the last
state in the cycles. From other states, 
the action $b$ leads to $q_{\bot}$ (transitions not depicted).
A sure winning strategy for eventually synchronization in $\{q_T\}$ is to 
play $a$ in the first $p^{\#}_n = \prod_{i=1}^{n} p_i$ steps, and then play $b$.
This requires memory of size $p^{\#}_n > 2^n$ while the size of $\M_n$
is in $O(n^2 \log n)$~\cite{BS96}.
It can be proved by standard pumping arguments that no strategy of size
smaller than $p^{\#}_n$ is sure winning.

\begin{figure}[t]
\begin{center}
\def\fsize{\normalsize}

\begin{picture}(93,65)(0,0)

{\fsize

\node[Nmarks=i, iangle=180](q0)(9,30){$q_0$}
\node[Nmarks=n](q1)(29,51){$q^1_1$}
\node[Nmarks=n](q2)(49,51){$q^1_2$}

\node[Nmarks=n](q3)(29,13){$q^2_1$}
\node[Nmarks=n](q4)(49,23){$q^2_2$}
\node[Nmarks=n](q5)(49,3){$q^2_3$}

\node[Nmarks=n](safe)(69,30){$q_T$}
\node[Nmarks=n](bad)(89,30){$q_{\bot}$}

\node[Nmarks=n, Nw=30, Nh=18, dash={1.5}0, ExtNL=y, NLangle=22, NLdist=1](A1)(39,51){$H_1$}
\node[Nmarks=n, Nw=30, Nh=32, dash={1.5}0, ExtNL=y, NLangle=38, NLdist=1](A2)(39,13){$H_2$}


\drawedge[ELpos=43, ELside=l, ELdist=1, curvedepth=0](q0,q1){$a,b: \sfrac{1}{2}$}
\drawedge[ELpos=40, ELside=r, ELdist=1, curvedepth=0](q0,q3){$a,b: \sfrac{1}{2}$}

\drawedge[ELpos=50, ELside=l, ELdist=1, curvedepth=4](q1,q2){$a$}
\drawedge[ELpos=50, ELside=l, ELdist=1, curvedepth=4](q2,q1){$a$}

\drawedge[ELpos=50, ELside=l, ELdist=1, curvedepth=4](q3,q4){$a$}
\drawedge[ELpos=50, ELside=r, ELdist=1.5, curvedepth=4](q4,q5){$a$}
\drawedge[ELpos=40, ELside=l, ELdist=1, curvedepth=4](q5,q3){$a$}

\drawedge[ELpos=50, ELside=l, ELdist=1, curvedepth=0](q2,safe){$b$}
\drawedge[ELpos=50, ELside=r, ELdist=1, curvedepth=0, syo=-3](q5,safe){$b$}


\drawedge[ELpos=48, ELside=l, ELdist=1, curvedepth=0](safe,bad){$a,b$}
\drawloop[ELside=r,loopCW=n, loopdiam=6, loopangle=90](bad){$a,b$}



}
\end{picture}
\caption{The MDP $\M_2$.}\label{fig:exp-mem}
\end{center}
\end{figure}

The following theorem summarizes the results for sure eventually synchronizing.

\begin{theorem}\label{theo:sure-eventually-pspace-c}
For sure eventually synchronizing  in MDPs:

\begin{enumerate}
\item (Complexity). The membership problem is PSPACE-complete.

\item (Memory). Exponential memory is necessary and sufficient for both pure 
and randomized strategies, and pure  strategies are sufficient. 
\end{enumerate}

\end{theorem}

\subsection{Almost-sure  eventually synchronization}\label{sec:almost-eventually}

We show an example where infinite memory is necessary to win for
almost-sure eventually synchronizing. Consider the MDP
in~\figurename~\ref{fig:inf-mem} with initial state $q_0$.
  We construct a strategy that is almost-sure
  eventually synchronizing in $q_2$, showing that $q_0 \in
  \winas{event}(q_2)$.  First, observe that for all $\epsilon > 0$ we
  can have probability at least $1 - \epsilon$ in $q_2$ after finitely
  many steps: playing 
  $n$ times $a$ and then
  $b$ leads to probability $1- \frac{1}{2^n}$ in $q_2$. Thus the MDP
  is limit-sure eventually synchronizing in $q_2$. Moreover the
  remaining probability mass is in~$q_0$.  It turns out that that from
  any (initial) distribution with support $\{q_0,q_2\}$, the MDP is
  again limit-sure eventually synchronizing in $q_2$, and with support
  in $\{q_0,q_2\}$. Therefore we can take a smaller value of
  $\epsilon$ and play a strategy to have probability at least $1 -
  \epsilon$ in $q_2$, and repeat this for $\epsilon \to 0$. This
  strategy ensures almost-sure eventually synchronizing in $q_2$.  The
  next result shows that infinite memory is necessary for almost-sure
  winning in this example.

\begin{lemma}\label{lem:inf-mmeory-almost-event}
There exists an almost-sure eventually synchronizing MDP for which 
all almost-sure eventually synchronizing strategies require infinite memory.
\end{lemma}

\begin{proof}
Consider the MDP $\M$ shown in~\figurename~\ref{fig:inf-mem}. We argued in 
Section~\ref{sec:almost-eventually} that $q_0 \in \winas{event}(q_2)$ 
and we now show that infinite memory is necessary from $q_0$ 
for almost-sure eventually synchronizing in $q_2$.

Assume towards contradiction that there exists a finite-memory
strategy $\alpha$ that is almost-sure eventually synchronizing in
$q_2$. Consider the Markov chain $\M(\alpha)$ (the product of the MDP
$\M$ with the finite-state transducer defining $\alpha$). A state
$(q,m)$ in $\M(\alpha)$ is called a \emph{$q$-state}. Since $\alpha$
is almost-sure eventually synchronizing (but is not sure eventually
synchronizing) in $q_2$ , there is a $q_2$-state in the recurrent states
of $\M(\alpha)$.  Since on all actions $q_0$ is a successor of $q_2$,
and $q_0$ is a successor of itself, it follows that there is a
recurrent $q_0$-state in $\M(\alpha)$, and that all periodic classes
of recurrent states in $\M(\alpha)$ contain a $q_0$-state. Hence, in
each stationary distribution there is a $q_0$-state with a positive
probability, and therefore the probability mass in $q_0$ is bounded
away from zero. It follows that the probability mass in $q_2$ is
bounded away from $1$ thus $\alpha$ is not almost-sure eventually
synchronizing in $q_2$, a contradiction.  \qed
\end{proof}

\begin{figure}[t]
\begin{center}
\begin{picture}(60,26)

\node[Nmarks=i,iangle=180](n0)(10,10){$q_0$}
\node[Nmarks=n](n1)(35,10){$q_1$}
\node[Nmarks=n](n2)(55,10){$q_2$}

\drawedge[ELdist=.5](n0,n1){$a,b: \sfrac{1}{2}$}
\drawloop[ELside=l,loopCW=y, loopangle=-90, loopdiam=4](n0){$a,b:\sfrac{1}{2}$}

\drawedge(n1,n2){$b$}
\drawloop[ELside=l,loopCW=y, loopangle=-90, loopdiam=4](n1){$a$}

\drawedge[ELpos=50, ELdist=.5, ELside=r, curvedepth=-10](n2,n0){$a,b$}

\end{picture}
\caption{An MDP where infinite memory is necessary for
almost-sure eventually synchronizing strategies.}\label{fig:inf-mem}
\end{center}
\end{figure}
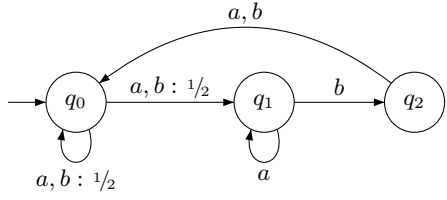


It turns out that in general, 
almost-sure eventually synchronizing strategies can be constructed from a family
of limit-sure eventually synchronizing strategies if we can also ensure that the
probability mass remains in the winning region (as in the MDP
in~\figurename~\ref{fig:inf-mem}).
We present a characterization of the winning region for almost-sure
winning based on an extension of the limit-sure eventually
synchronizing objective \emph{with exact support}.  This objective
requires to ensure probability arbitrarily close to $1$ in the target
set $T$, and moreover that after the same number of
steps the support of the probability distribution is contained in a
given set $U$.  Formally, given an MDP $\M$, let
$\winlim{event}(\fsum_T,U)$ for $T \subseteq U$ be the set of all
initial distributions such that for all $\epsilon>0$ there exists a
strategy~$\alpha$ and $n \in \nat$ such that $\M^{\alpha}_n(T) \geq
1-\epsilon$ and $\M^{\alpha}_n(U)=1$.
We say that $\alpha$ is limit-sure eventually synchronizing in $T$ with support in $U$.

We will present an algorithmic solution to limit-sure eventually
synchronizing objectives with exact support in Section~\ref{sec:limit-sure}.
Our characterization of the winning region for almost-sure winning is
as follows.


\begin{lemma}\label{lem: almost-limit-reduce-limit-event}
Let $\M$ be an MDP and $T$ be a target set. For all states $q_0$, 
we have $q_0\in \winas{event}(\fsum_T)$
if and only if
there exists a set~$U$ such that:
  \begin{itemize}
    \item $q_0 \in \winsure{event}(\fsum_U)$, and \smallskip
    \item $d_U \in \winlim{event}(\fsum_{T},U)$ where $d_U$ is the uniform distribution over~$U$.
  \end{itemize}
\end{lemma}

\begin{proof}
First, if $q_0 \in \winas{event}(\fsum_T)$,
then there is a strategy $\alpha$ such that $\sup_{n \in \nat} \M^{\alpha}_n(T)=1$.
Then either $\M^{\alpha}_n(T) = 1$ for some $n \geq 0$,
or $\limsup_{n\to \infty} \M^{\alpha}_n(T)=1$. 
If $\M^{\alpha}_n(T) = 1$, then $q_0$ is sure winning 
for eventually synchronizing in $T$, thus $q_0 \in \winsure{event}(\fsum_T)$
and we can take $U = T$. 
Otherwise,
for all $i>0$ there exists $n_i \in \nat$ such that 
$\M^{\alpha}_{n_i}(T) \geq 1-2^{-i}$, and moreover $n_{i+1} > n_i$ for all $i>0$.
Let $s_i = \Supp(\M^{\alpha}_{n_i})$ be the support of~$\M^{\alpha}_{n_i}$.
Since the state space is finite, there is a set~$U$ that 
occurs infinitely often in the sequence~$s_0 s_1 \dots$,
thus for all $k>0$ there exists $m_k \in \nat$ such that 
$\M^{\alpha}_{m_k}(T) \geq 1-2^{-k}$ and 
$\M^{\alpha}_{m_k}(U) = 1$.
It follows that $\alpha$ is sure eventually synchronizing in~$U$ 
from $q_0$, hence $q_0 \in \winsure{event}(\fsum_U)$.
Moreover $\M$ with intial distribution $d_1 = \M^{\alpha}_{m_1}$
is limit-sure eventually synchronizing in $T$ with exact support in $U$.
Since $\Supp(d_1) = U = \Supp(d_U)$, it follows 
by Corollary~\ref{col:uniform-dist-limit} that 
$d_U \in \winlim{event}(\fsum_{T},U)$.

To establish the converse, note that since $d_U \in
\winlim{event}(\fsum_{T},U)$, it follows from
Corollary~\ref{col:uniform-dist-limit} that from all initial
distributions with support in $U$, for all $\epsilon > 0$ there exists
a strategy $\alpha_{\epsilon}$ and a position $n_{\epsilon}$ such that
$\M^{\alpha_{\epsilon}}_{n_{\epsilon}}(T) \geq 1-\epsilon$ and
$\M^{\alpha_{\epsilon}}_{n_{\epsilon}}(U) = 1$.  We construct an
almost-sure limit eventually synchronizing strategy $\alpha$ as
follows.  Since $q_0 \in \winsure{event}(\fsum_U)$, play according to
a sure eventually synchronizing strategy from $q_0$ until all the
probability mass is in $U$.  Then for $i=1,2, \dots$ and $\epsilon_i =
2^{-i}$, repeat the following procedure: given the current probability distribution, 
select the corresponding strategy $\alpha_{\epsilon_i}$
and play according to $\alpha_{\epsilon_i}$ for $n_{\epsilon_i}$
steps, ensuring probability mass at least $1-2^{-i}$ in $T$, and since
after that the support of the probability mass is again in $U$,
%
play according to $\alpha_{\epsilon_{i+1}}$ for $n_{\epsilon_{i+1}}$
steps, etc.  This strategy $\alpha$ ensures that $\sup_{n\in\nat}
\M^{\alpha}_n(T)=1$ from $q_0$, hence $q_0\in \winas{event}(\fsum_T)$.
\qed
\end{proof}

Note that from Lemma~\ref{lem: almost-limit-reduce-limit-event}, it follows
that counting strategies are sufficient
to win almost-sure eventually synchronizing objective
(a strategy is \emph{counting} if $\alpha(\rho) = \alpha(\rho')$
for all prefixes $\rho, \rho'$ with the same length and $\Last(\rho) = \Last(\rho')$).

As we show in Section~\ref{sec:limit-sure} that the membership problem
for limit-sure eventually synchronizing with exact support can be solved
in PSPACE, it follows from the characterization in Lemma~\ref{lem: almost-limit-reduce-limit-event} 
that the membership problem for almost-sure eventually synchronizing
is in PSPACE, using the following (N)PSPACE algorithm: 
guess the set $U$, and check that $q_0 \in \winsure{event}(\fsum_U)$, and that 
$d_U \in \winlim{event}(\fsum_{T},U)$
where $d_U$ is the uniform distribution over~$U$ (this can be done
in PSPACE by Theorem~\ref{theo:sure-eventually-pspace-c}
and Theorem~\ref{theo:limit-sure-eventually}).
We present a matching lower bound.
 

\begin{lemma}
The membership problem for $\winas{event}(\fsum_T)$ is PSPACE-hard
even if $T$ is a singleton.
\end{lemma}

\begin{proof}
The proof is by a reduction from the membership problem for sure eventually
synchronization, which is PSPACE-complete by Theorem~\ref{theo:sure-eventually-pspace-c}.
Given an MDP $\M=\tuple{Q, \Act,\delta}$, an initial state $q_0 \in Q$,
and a state $\q \in Q$, we construct an MDP $\N=\tuple{Q',\Act',\delta'}$ and
a state $\p \in Q'$ such that $q_0 \in \winsure{event}(\q)$ in $\M$
if and only if $q_0 \in \winas{event}(\p)$ in $\N$.
The MDP $\N$ is a copy
of $\M$ with two new states $\p$ and $\sink$ reachable only by a new action $\sharp$
(see \figurename~\ref{fig:almost-ss-reduction}).
Formally, $Q' = Q \cup \{\p, \sink\}$ and $\Act'= \Act \cup \{\sharp\}$,
and the transition function $\delta'$ is defined as follows, for all $q \in Q$:
$\delta'(q,a) = \delta(q,a)$ for all $a \in \Act$,
$\delta'(q,\sharp)(\sink) = 1$ if $q \neq \q$, and $\delta'(\q,\sharp)(\p) = 1$;
finally, for all $a \in \Act'$, let $\delta'(\p,a)(\sink) = \delta'(\sink,a)(\sink) = 1$.

\begin{figure}[t]
\begin{center}
    \begin{picture}(115,36)(0,2)

\node[Nmarks=n, Nw=40, Nh=14, dash={0.2 0.5}0](m1)(20,24){}
\node[Nframe=n](label)(8,28){MDP $\M$}
\node[Nmarks=n](n1)(28,22){$\q$}
\node[Nmarks=n](n2)(10,22){$q$}
\node[Nframe=n](arrow)(45,22){{\Large $\Rightarrow$}}

\node[Nmarks=n, Nw=40, Nh=14, dash={0.2 0.5}0](nm1)(80,24){}
\node[Nmarks=n, Nw=50, Nh=36, dash={0.4 1}0](m2)(80,20){}
\node[Nframe=n](label)(65,34){MDP $\N$}

\node[Nframe=n](label)(68,28){MDP $\M$}
\node[Nmarks=n](n1)(88,22){$\q$}
\node[Nmarks=n](n2)(70,22){$q$}

\node[Nmarks=n](end)(70,8){$\sink$}
\node[Nmarks=n](qq)(88,8){$\p$} 

\drawloop[ELside=l,loopCW=y, loopangle=180, loopdiam=4](end){$\Act'$}

\drawedge[ELpos=50, ELside=l, curvedepth=0](n2,end){$\sharp$}
\drawedge[ELpos=50, ELside=r, curvedepth=0](n1,qq){$\sharp$}
\drawedge(qq,end){$\Act'$}

\end{picture}
\end{center}
 \caption{Sketch of the reduction to show PSPACE-hardness of 
 the membership problem for almost-sure eventually synchronizing.}\label{fig:almost-ss-reduction}
\end{figure}
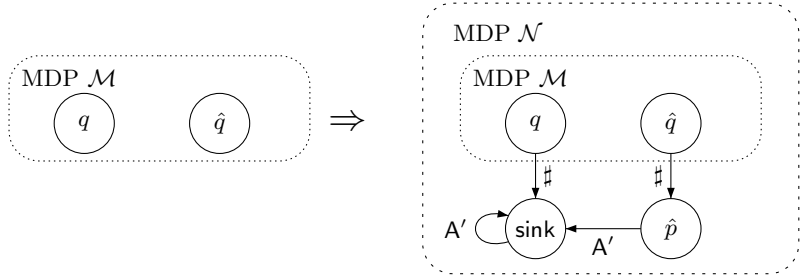

The goal is that $\N$ simulates $\M$ until the action $\sharp$ is played in $\q$
to move the probability mass from $\q$ to $\p$,
ensuring that if $\M$ is sure-winning for eventually synchronizing in $\q$, then
$\N$ is also sure-winning (and thus almost-sure winning) for eventually synchronizing in $\p$. Moreover,
the only way to be almost-sure eventually synchronizing in $\p$ is to have
probability~$1$ in $\p$ at some point, because the state $\p$ is transient
under all strategies, thus the probability mass cannot accumulate and tend to $1$
in $\p$ in the long run. It follows that (from all initial states
$q_0$) $\M$ is sure-winning for eventually
synchronizing in $\q$ if and only if $\N$ is almost-sure winning
for eventually synchronizing in $\p$. It follows from this reduction that
the membership problem for almost-sure eventually synchronizing objective
is PSPACE-hard.
\qed
\end{proof}




\noindent The results of this section are summarized as follows.
 
\begin{theorem}\label{theo:almost-sure-eventually}
For almost-sure eventually synchronizing  in MDPs:

\begin{enumerate}
\item (Complexity). The membership problem is PSPACE-complete.

\item (Memory). Infinite memory is necessary in general for both pure 
and randomized strategies, and pure strategies are sufficient.
\end{enumerate}

\end{theorem}


\subsection{Limit-sure eventually synchronization}  \label{sec:limit-sure}


In this section, we present the algorithmic solution for 
limit-sure eventually synchronizing with exact support.
Note that the limit-sure eventually synchronizing objective is
a special case where the support is the state space of the MDP.
Consider the MDP in \figurename~\ref{fig:almost-limit-eventually-differ}
which is limit-sure eventually synchronizing in $\{q_2\}$, 
as shown in Lemma~\ref{lem:dif-in-def}. For $i=0,1,\dots$, the sequence $\Pre^i(T)$ 
of predecessors of $T = \{q_2\}$ is ultimately periodic: 
$\Pre^0(T) = \{q_2\}$, and $\Pre^i(T) = \{q_1\}$ for all $i \geq 1$. 
Given $\epsilon > 0$, a strategy to get probability $1-\epsilon$ in $q_2$
first accumulates probability mass in the \emph{periodic} subsequence of predecessors (here $\{ q_1 \}$),
and when the probability mass is greater than $1-\epsilon$ in $q_1$, the
strategy injects the probability mass in $q_2$ (through the aperiodic prefix of
the sequence of predecessors). This is the typical shape of a limit-sure eventually 
synchronizing strategy. Note that in this scenario, the MDP is also limit-sure eventually synchronizing 
in every set $\Pre^i(T)$ of the sequence of predecessors.
A special case is when it is possible to get probability~$1$
in the sequence of predecessors after finitely many steps. 
In this case, the probability mass injected in $T$ is $1$ and the MDP is even sure-winning. 
The algorithm for deciding limit-sure eventually synchronization relies
on the above characterization, generalized in Lemma~\ref{lem:lse-pre} 
to limit-sure eventually synchronizing with exact support, saying that 
limit-sure eventually synchronizing in~$T$ with support in~$U$ 
is equivalent to either limit-sure eventually synchronizing in $\Pre^k(T)$ with support in $\Pre^k(U)$
(for arbitrary~$k$), or sure eventually synchronizing in~$T$ (and therefore also in~$U$).

\begin{lemma} \label{lem:lse-pre}

For $T \subseteq U$ and $k\geq 0$, let $R = \Pre^{k}(T)$ and $Z = \Pre^{k}(U)$.  
Then,
$\winlim{event}(\fsum_T,U) = \winsure{event}(\fsum_T) \cup
\winlim{event}(\fsum_R,Z)$.
\end{lemma}

\begin{proof}
The proof is in two parts. First we show that 
$\winsure{event}(\fsum_T) \cup \winlim{event}(\fsum_R,Z) \subseteq \winlim{event}(\fsum_T,U)$:
since $T\subseteq U$, it follows from the definitions that 
$\winsure{event}(\fsum_T) \subseteq \winlim{event}(\fsum_T,U)$;
to show that $\winlim{event}(\fsum_R, Z) \subseteq \winlim{event}(\fsum_T, U)$
in an MDP $\M$,
let $\epsilon > 0 $ and
consider an initial distribution $\mu_0$ and a strategy $\alpha$ such that for 
some $i \geq 0$ we have $\M^{\alpha}_{i}(R)\geq 1-\epsilon$ and $\M^{\alpha}_i(Z) = 1$.
We construct a strategy~$\beta$ that plays like~$\alpha$ for the first~$i$
steps, and then since $R = \Pre^{k}(T)$ and $Z = \Pre^{k}(U)$ plays
from states in~$R$ according to a sure eventually synchronizing strategy with target~$T$,
and from states in $Z \setminus R$ according to a sure eventually synchronizing strategy with target~$U$
(such strategies exist by the proof of Lemma~\ref{lem:sure-ss-pre}).
The strategy $\beta$ ensures from $\mu_0$ that $\M^{\beta}_{i+k}(T) \geq 1-\epsilon$ 
and $\M^{\beta}_{i+k}(U) = 1$, showing that $\M$ is limit-sure eventually synchronizing 
in~$T$ with support in~$U$.

Second we show the converse inclusion, namely that 
$\winlim{event}(\fsum_T,U) \subseteq \winsure{event}(\fsum_T) \cup \winlim{event}(\fsum_R,Z)$.
Consider an initial distribution $\mu_0 \in \winlim{event}(\fsum_T,U)$ in the MDP $\M$
and for $\epsilon_i = \frac{1}{i}$ ($i \in \nat$) let $\alpha_i$ be a strategy 
and $n_i \in \nat$ such that $\M^{\alpha_i}_{n_i}(T)\geq 1-\epsilon_i$
 and $\M^{\alpha_i}_{n_i}(U) = 1$.
We consider two cases. $(a)$ If the set $\{n_i \mid i \geq 0\}$ is bounded,
then there exists a number $n$ that occurs infinitely often in the sequence
$(n_i)_{i \in \nat}$, and such that for all $i \geq 0$, there exists
a strategy $\beta_i$ such 
that $\M^{\beta_i}_{n}(T) \geq 1-\epsilon$ and $\M^{\beta_i}_{n}(U) = 1$.
Since $n$ is fixed, we can assume w.l.o.g. that the strategies $\beta_i$
are pure, and since there is a finite number of pure strategies
over paths of length at most $n$, it follows that there is a strategy $\beta$
that occurs infinitely often among the strategies $\beta_i$ and
such that for all $\epsilon > 0$ we have 
$\M^{\beta}_{n}(T) \geq 1-\epsilon$, hence $\M^{\beta}_{n}(T) = 1$,
showing that $\M$ is sure winning for eventually synchronizing in $T$, 
that is $\mu_0 \in \winsure{event}(\fsum_T)$.
$(b)$ otherwise, the set $\{n_i \mid i \geq 0\}$ is unbounded
and we can assume w.l.o.g. that $n_i \geq k$ for all $i \geq 0$.
We claim that the family of strategies $\alpha_i$ ensures 
limit-sure synchronization in $R = \Pre^{k}(T)$ with 
support in $Z = \Pre^{k}(U)$. Essentially this is because if the probability
in $T$ is close to $1$ after $n_i$ steps, then $k$ steps before 
the probability in $\Pre^{k}(T)$ must be close to $1$ as well.
%
%
Formally, we show that $\alpha_i$ is such that 
$\M^{\alpha_i}_{n_i - k}(R)\geq 1- \frac{\epsilon}{\eta^k}$ and $\M^{\alpha_i}_{n_i - k}(Z) = 1$
where $\eta$ is the smallest positive probability in the transitions of $\M$.
Towards contradiction, assume that $\M^{\alpha_i}_{n_i - k}(R)< 1- \frac{\epsilon}{\eta^k}$,
then $\M^{\alpha_i}_{n_i - k}(Q \setminus R) > \frac{\epsilon}{\eta^k}$ and from every
state $q \in Q \setminus R$, no matter which sequence of actions is played
by $\alpha_i$ for the next $k$ steps, there is a path from $q$ to a state
outside of $T$, thus with probability at least $\eta^k$. Hence the probability
in $Q \setminus T$ after $n_i$ steps is greater than 
$\frac{\epsilon}{\eta^k} \cdot \eta^k$, and it follows that 
$\M^{\alpha_i}_{n_i}(T)< 1-\epsilon$, in contradiction with the definition of $\alpha_i$.
This shows that $\M^{\alpha_i}_{n_i - k}(R)\geq 1- \frac{\epsilon}{\eta^k}$,
and an argument analogous to the proof of Lemma~\ref{lem:sure-ss-pre}
shows that $\M^{\alpha_i}_{n_i - k}(Z) = 1$.
It follows that $\mu_0 \in \winlim{event}(\fsum_R,Z)$ and the proof is complete.
\qed 
\end{proof}

Thanks to Lemma~\ref{lem:lse-pre}, since sure-winning is already
solved in Section~\ref{sec:sure-event-sync}, it suffices to solve the
limit-sure eventually synchronizing problem for target $R = \Pre^k(T)$
and support $Z = \Pre^k(U)$ with arbitrary $k$, instead of $T$ and
$U$. We can choose $k$ such that both $\Pre^k(T)$ and $\Pre^k(U)$ lie in
the periodic part of the sequence of pairs of predecessors
$(\Pre^i(T),\Pre^i(U))$.  We can assume that $k \leq 3^{\abs{Q}}$
since $\Pre^i(T) \subseteq \Pre^i(U) \subseteq Q$ for all $i \geq 0$.
For such value of $k$ the limit-sure problem is conceptually simpler:
once some probability is injected in $R = \Pre^k(T)$, it can loop
through the sequence of predecessors and visit $R$ infinitely often
(every $r$ steps, where $r \leq 3^{\abs{Q}}$ is the period of the
sequence of pairs of predecessors).
It follows that if a strategy ensures with probability $1$
that the set $R$ can be reached by finite paths whose lengths are 
congruent modulo $r$, then the whole probability mass can 
indeed synchronously accumulate in $R$ in the limit.

Therefore, limit-sure eventually synchronizing in $R$ 
reduces to standard limit-sure reachability with target set $R$ and the additional 
requirement that the numbers of steps at which the target set is reached 
be congruent modulo $r$.
In the case of limit-sure eventually synchronizing with support in~$Z$, we also need to 
ensure that no mass of probability leaves the sequence $\Pre^i(Z)$. 
In a state $q \in \Pre^i(Z)$, we say that an action $a \in \Act$ is \emph{$Z$-safe}
at position $i$ if\footnote{Since $\Pre^r(Z) = Z$ and $\Pre^r(R) = R$, 
we assume a modular arithmetic for exponents of $\Pre$, that is $\Pre^x(\cdot)$ 
is defined as $\Pre^{x \!\!\mod r}(\cdot)$. For example $\Pre^{-1}(Z)$ is $\Pre^{r-1}(Z)$.}
$\post(q,a)\subseteq \Pre^{i-1}(Z)$. In states 
$q \not\in \Pre^i(Z)$ there is no $Z$-safe action at position~$i$.

To encode the above requirements, we construct an MDP $\M_Z \times[r]$ 
that allows only $Z$-safe actions to be played (and then mimics the original
MDP), and tracks the position (modulo $r$) in the sequence of predecessors,
thus simply decrementing the position on each transition since all successors
of a state $q \in \Pre^i(Z)$ on a safe action are in $\Pre^{i-1}(Z)$.

Formally, if $\M = \tuple{Q, \Act, \delta}$ then 
$\M_Z \times[r] = \tuple{Q', \Act, \delta'}$ where 
\begin{itemize}
\item $Q' = Q \times \{r-1, \dots, 1, 0\} \cup \{\sink\}$; 
intuitively, we expect that $q \in \Pre^i(Z)$ in the reachable states $\tuple{q,i}$ 
consisting of a state $q$ of $\M$ and a \emph{position}~$i$ in the predecessor sequence;

\item $\delta'$ is defined as follows (assuming an arithmetic modulo $r$ on positions) 
      for all $\tuple{q,i} \in Q'$ and $a \in \Act$:
      if $a$ is a $Z$-safe action in $q$ at position $i$, 
      then $\delta'(\tuple{q,i},a)(\tuple{q',i-1}) = \delta(q,a)(q')$, 
      otherwise $\delta'(\tuple{q,i},a)(\sink) = 1$ 
      (and $\sink$ is absorbing).
\end{itemize}

Note that the size of the MDP $\M_Z \times[r]$ is exponential in the size of $\M$
(since $r$ is at most $3^{\abs{Q}}$).

\begin{lemma} \label{lem:lssr-assr}
Let $\M$ be an MDP and $R \subseteq Z$ be two sets of states such that   
$\Pre^{r}(R)=R$ and $\Pre^{r}(Z)=Z$ where~$r>0$. 
Then a state $q_0$ is limit-sure eventually synchronizing in $R$ with support in $Z$ 
($q_0 \in \winlim{event}(\fsum_R, Z)$) if and only if  
there exists $0 \leq t < r$ such that $\tuple{q_0,t}$ 
is limit-sure winning for the reachability objective 
$\Diamond (R \times \{0\})$ in the MDP $\M_Z \times [r]$.
\end{lemma}

\begin{proof}
For the first direction of the lemma, assume
that $q_0$ is limit-sure eventually synchronizing in $R$ with support in
$Z$, and for $\epsilon >0$ let $\beta$ be a strategy such that
$\M^{\beta}_{k}(Z)=1$ and $\M^{\beta}_k(R) \geq 1-\epsilon$ for some
number $k$ of steps.  Let $0 \leq t \leq r$ such that $t=k \mod r$.
We show that from initial state $(q_0, t)$ the strategy $\alpha$ in $\M_Z \times [r]$ that
mimics (copies) the strategy $\beta$ is limit-sure winning for the
reachability objective $\Diamond R_0$: it follows from
Lemma~\ref{lem:sure-ss-pre} that $\alpha$ plays only $Z$-safe actions,
and since $Pr^{\alpha}(\Diamond R_0)\geq Pr^{\alpha}(\Diamond^k R_0) =
\M^{\beta}_k(R)\geq 1-\epsilon$, the result follows. 

For the converse direction, let $R_0 = R \times \{0\}$ and assuming
that there exists $0 \leq t < r$ such that $\tuple{q_0,t}$ is
limit-sure winning for the reachability objective $\Diamond R_0$ in
$\M_Z\times[r]$, show that $q_0$ is limit-sure synchronizing in
target set $R$ with exact support in $Z$.  
Since the winning region of limit-sure and almost-sure reachability coincide 
for MDPs~\cite{AHK07}, there exists a (pure) strategy $\alpha$ in 
$\M_Z \times [r]$ with initial state $\tuple{q,t}$ such that
$\Pr^{\alpha}(\Diamond R_0) = 1$.

Given $\epsilon>0$, we construct from $\alpha$ a pure strategy $\beta$ in $\M$ 
that is $(1-\epsilon)$-synchronizing in $R$ with support in~$Z$. 
Given a finite path $\rho = q_0 a_0 q_1 a_1 \dots q_n$ in $\M$ (with $q_0 = q$), there 
is a corresponding path $\rho' = \tuple{q_0,k_0} a_0 \tuple{q_1,k_1} a_1 \dots \tuple{q_n, k_n}$ in 
$\M_Z \times [r]$ where $k_0 = t$ and $k_{i+1} = k_i - 1$ for all $i \geq 0$.
Since the sequence $k_0, k_1, \dots$ is uniquely determined from $\rho$, 
there is a clear bijection between the paths in $\M$ and the paths in $\M_Z \times [r]$
that we often omit to apply and mention.
Define the strategy $\beta$ as follows:
if $q_n \in \Pre^{k_n}(R)$, then there exists an action $a$ such that 
$\post(q_n,a) \subseteq \Pre^{k_n-1}(R)$ and we define $\beta(\rho) = a$, 
otherwise let $\beta(\rho) = \alpha(\rho')$.
%
Thus $\beta$ mimics $\alpha$ (thus playing only $Z$-safe actions) unless a 
state $q$ is reached at step $n$ such that $q \in \Pre^{t-n}(R)$, 
and then $\beta$ switches to always playing actions
that are $R$-safe (and thus also $Z$-safe since $R \subseteq Z$).
%
%
We now prove that $\beta$ is limit-sure eventually synchronizing in
target set $R$ with support in $Z$. First since $\beta$ plays only
$Z$-safe actions, it follows for all $k$ such that $t-k = 0$ (modulo $r$), all
states reached from $q_0$ with positive probability after $k$ steps are in $Z$. 
Hence $\M^{\beta}_{k}(Z)=1$ for all such $k$. 
%
%
Second, we show that given $\epsilon>0$ there exists $k$ such that
$t-k = 0$ and $\M^{\beta}_{k}(R) \geq 1-\epsilon$, thus also
$\M^{\beta}_{k}(Z)=1$ and $\beta$ is limit-sure eventually
synchronizing in target set $R$ with support in~$Z$.  To show this,
recall that $\Pr^{\alpha}(\Diamond R_0) = 1$, and therefore
$\Pr^{\alpha}(\Diamond^{\leq k} R_0) \geq 1-\epsilon$ for all
sufficiently large $k$. Without loss of generality, consider such a
$k$ satisfying $t-k = 0$ (modulo $r$). For $i = 1, \dots, r-1$, let
$R_i = \Pre^i(R) \times \{i\}$.  Then trivially
$\Pr^{\alpha}(\Diamond^{\leq k} \bigcup_{i=0}^{r} R_i) \geq
1-\epsilon$ and since $\beta$ agrees with $\alpha$ on all finite paths
that do not (yet) visit $\bigcup_{i=0}^{r} R_i$, given a path $\rho$
that visits $\bigcup_{i=0}^{r} R_i$ (for the first time), only
$R$-safe actions will be played by $\beta$ and thus all continuations
of $\rho$ in the outcome of $\beta$ will visit $R$ after $k$ steps (in
total).  It follows that $\Pr^{\beta}(\Diamond^{=k} R_0) \geq
1-\epsilon$, that is $\M^{\beta}_k(R) \geq 1-\epsilon$. Note that we
used the same strategy $\beta$ for all $\epsilon >0$ and thus $\beta$
is also almost-sure eventually synchronizing in $R$.
 \qed
\end{proof}


Since deciding limit-sure reachability is PTIME-complete, it follows from
Lemma~\ref{lem:lssr-assr} that limit-sure synchronization (with exact support)
can be decided in EXPTIME. We show that the problem can be solved in PSPACE
by exploiting the special structure of the exponential MDP in Lemma~\ref{lem:lssr-assr}.
We conclude this section by showing that limit-sure synchronization with exact 
support is PSPACE-complete (even in the special case of a trivial support).

\begin{lemma}\label{lem:pspace-mi}
The membership problem for limit-sure eventually synchronization with exact support
is in PSPACE.
\end{lemma}

\begin{proof}
We present a (nondeterministic) PSPACE algorithm to decide, 
given an MDP $\M=\tuple{Q,\Act,\delta}$, a state $q_0$, and two sets $T \subseteq U$, 
whether $q_0$ is limit-sure eventually synchronizing in~$T$ with support in~$U$. 

First, the algorithm computes numbers $k\geq 0$ and $r>0$ such that
for $R = \Pre^k(T)$ and $Z = \Pre^k(U)$ we have $\Pre^{r}(R)=R$ and
$\Pre^{r}(Z)=Z$. As discussed before, this can be done by guessing
$k,r \leq 3^{\abs{Q}}$.  By Lemma~\ref{lem:lse-pre}, we have
$\winlim{event}(\fsum_T,U) = \winlim{event}(\fsum_R,Z) \cup
\winsure{event}(\fsum_T)$, and since sure eventually synchronizing in
$T$ can be decided in PSPACE (by
Theorem~\ref{theo:sure-eventually-pspace-c}), it suffices to decide
limit-sure eventually synchronizing in $R$ with support in $Z$ in
PSPACE.  According to Lemma~\ref{lem:lssr-assr}, it is therefore
sufficient to show that deciding limit-sure winning for the (standard)
reachability objective $\Diamond (R \times \{0\})$ in the MDP $\M_Z
\times [r]$ can be done in polynomial space.  As we cannot afford to
construct the exponential-size MDP $\M_Z \times [r]$, the algorithm
relies on the following characterization of the limit-sure winning set
for reachability objectives in MDPs. It is known that the winning
region for limit-sure and almost-sure reachability
coincide~\cite{AHK07}, and pure memoryless strategies are sufficient.
Therefore, we can see that the almost-sure winning set $W$ for the
reachability objective $\Diamond (R \times \{0\})$ satisfies the
following property: 
there exists a
memoryless strategy $\alpha: W \to \Act$ such that $(1)$ $W$ is
closed, that is $\post(q,\alpha(q)) \subseteq W$ for all $q \in W$,
and $(2)$ in the graph of the Markov chain $M(\alpha)$, for every
state $q \in W$, there is a path (of length at most $\abs{W}$) from
$q$ to $R \times \{0\}$.

This property ensures that from every state in $W$, the target set $R
\times \{0\}$ is reached within $\abs{W}$ steps with positive (and
bounded) probability, and since $W$ is closed it ensures that $R
\times \{0\}$ is reached with probability~$1$ in the long run.  Thus
any set $W$ satisfying the above property is almost-sure winning.

Our algorithm will guess and explore on the fly a set $W$ to ensure
that it satisfies this property, and contains the state
$\tuple{q_0,t}$ for some $t < r$.  As we cannot afford to explicitly
guess $W$ (remember that $W$ could be of exponential size), we
decompose $W$ into \emph{slices} $W_0, W_1,\dots$ such that $W_i
\subseteq Q$ and $W_i \times \{-i \mod r\} = W \cap (Q \times \{-i
\mod r\})$.  We start by guessing $W_0$, and we use the property 
that in $\M_Z \times [r]$, from a state
$(q,j)$ under all $Z$-safe actions, all successors are of the form
$(\cdot, j-1)$.  It follows that the successors of the states in~$W_i
\times \{-i\}$ should lie in the slice $W_{i+1} \times \{-i-1\}$, and
we can guess on the fly the next slice $W_{i+1} \subseteq Q$ by
guessing for each state $q$ in a slice $W_i$ an action $a_q$ such that
$\bigcup_{q \in W_i} \post(q,a_q) \subseteq W_{i+1}$.  Moreover, we
need to check the existence of a path from every state in $W$ to $R
\times \{0\}$. As $W$ is closed, it is
sufficient to check that there is a path from every state in $W_0
\times \{0\}$ to $R \times \{0\}$.  To do this we guess along with the
slices $W_0, W_1, \dots$ a sequence of sets $P_0, P_1, \dots$ where
$P_i \subseteq W_i$ contains the states of slice $W_i$ that belong to
the guessed paths. Formally, $P_0 = W_0$, and for all $i\geq 0$, the
set $P_{i+1}$ is such that $\post(q,a_q) \cap P_{i+1} \neq \emptyset$
for all $q \in P'_i$ (where $P'_i = P_i \setminus R$ if $i$ is a
multiple of $r$, and $P'_i = P_i$ otherwise), that is $P_{i+1}$
contains a successor of every state in $P_i$ that is not already in
the target $R$ (at position $0$ modulo $r$).

We need polynomial space to store the first slice $W_0$, 
the current slice $W_i$ and the set $P_i$, and the value of $i$ (in binary).
As $\M_Z \times [r]$ has $\abs{Q}\cdot r$ states, the algorithm runs 
for $\abs{Q}\cdot r$ iterations and then checks that 
$(1)$ $W_{\abs{Q}\cdot r} \subseteq W_0$ to ensure that 
$W = \bigcup_{i \leq \abs{Q}\cdot r} W_i\times\{i \mod r\}$ is closed, 
$(2)$ $P_{\abs{Q}\cdot r} = \emptyset$ showing that from every state in $W_0 \times \{0\}$ 
there is a path to $R \times \{0\}$ (and thus also from all states in $W$), 
and $(3)$ the state $q_0$ occurs in some slice $W_i$. 
The correctness of the algorithm follows from the characterization of the almost-sure 
winning set for reachability in MDPs: if some state $\tuple{q_0,t}$ is limit-sure winning, then 
the algorithm accepts by guessing (slice by slice) the almost-sure winning set $W$ 
and the paths from $W_0 \times \{0\}$ to $R \times \{0\}$ (at position $0$ modulo $r$), 
and otherwise any set (and paths) correctly 
guessed by the algorithm would not contain $q_0$ in any slice.

\qed
\end{proof}


It follows from the proof of Lemma~\ref{lem:lssr-assr} that
all winning modes for eventually synchronizing
are independent of the numerical value of the positive transition probabilities.

\begin{corollary}\label{col:uniform-dist-limit}
Let $\mu\in \{sure, almost, limit\}$ and $T\subseteq U$ be two sets. 
For two  distributions $d,d' $ with $\Supp(d)=\Supp(d')$, we have 
$d \in \win{event}{\mu}(\fsum_T,U)$~if and only if~$d' \in \win{event}{\mu}(\fsum_T,U)$.
\end{corollary}

To establish the PSPACE-hardness for limit-sure eventually synchronizing in MDPs, 
we use a reduction from the universal finiteness problem for 1L-AFAs.

\begin{lemma}\label{lem:limit-event-pspace-hard}
The membership problem for $\winlim{event}(\fsum_T)$ is PSPACE-hard
even if $T$ is a singleton.
\end{lemma}

\begin{figure}[t]
\begin{center}
\begin{picture}(115,46)(0,2)

\node[Nmarks=n, Nw=40, Nh=22, dash={0.2 0.5}0](m1)(20,20){}
\node[Nframe=n](label)(8,13){MDP $\M$}
\drawpolygon[dash={0.8 0.5}0](38,30)(23,30)(23,10)(38,10)
\node[Nframe=n](label)(30,13){$T\subseteq Q$}
\node[Nmarks=n](n1)(30,20){$q_2$}

\node[Nframe=n](label)(19,20){$\dots$}
\node[Nmarks=n](n2)(10,20){$q_1$}
\node[Nframe=n](arrow)(45,20){{\Large $\Rightarrow$}}

\node[Nmarks=n, Nw=40, Nh=22, dash={0.2 0.5}0](nm1)(80,20){}
\node[Nmarks=n, Nw=50, Nh=48, dash={0.4 1}0](m2)(80,32){}
\node[Nframe=n](label)(65,52){MDP $\N$}

\node[Nframe=n](label)(68,13){MDP $\M$}
\drawpolygon[dash={0.8 0.5}0](98,30)(83,30)(83,10)(98,10)
\node[Nframe=n](label)(90,13){$T\subseteq Q$}
\node[Nmarks=n](nn1)(90,20){$q_2$}
\node[Nframe=n](label)(79,20){$\dots$}
\node[Nmarks=n](nn2)(70,20){$q_1$}

\node[Nmarks=n](qq)(80,43){$q_{\init}$} 

\drawloop[ELside=l,loopCW=y, loopangle=90, loopdiam=4](qq){$\Act$}

\drawedge[ELpos=50, ELside=l](qq,nn1){$\sharp$}
\node[Nframe=n](label)(80.2,36){$\cdots$}
\drawedge[ELpos=50, ELside=r](qq,nn2){$\sharp$}

\drawedge[ELpos=50, ELside=r, curvedepth=-8](nn1,qq){$\sharp$}

\drawedge[ELpos=50, ELside=l, curvedepth=+8](nn2,qq){$\sharp$}

\end{picture}
\caption{Sketch of reduction to show PSPACE-hardness of  the membership problem
for limit-sure eventually synchronizing.}\label{fig:lim-sure-reduction}
\end{center}
\end{figure}
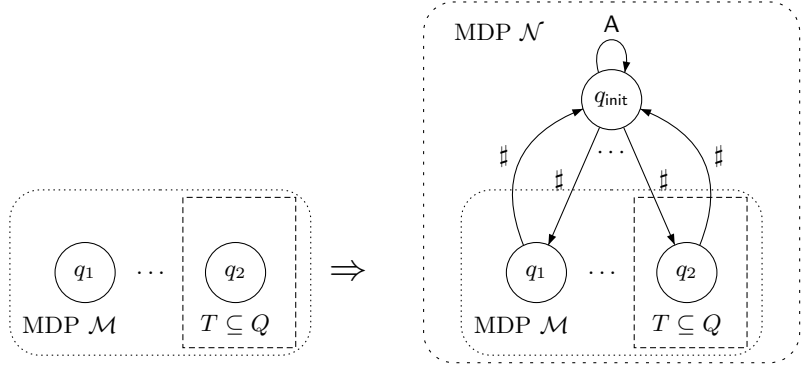

\begin{proof}
The proof is by a reduction from the universal finiteness problem for one-letter
alternating automata (1L-AFA), which is PSPACE-complete (by Lemma~\ref{lem:universal-finiteness-pspace-hard}).
It is easy to see that this problem remains PSPACE-complete even if the set $T$ of accepting
states of the 1L-AFA is a singleton, and 
given the tight relation between 1L-AFA and MDP (see Section~\ref{sec:1L-AFA}),
it follows from the definition of the universal finiteness problem that 
deciding, in an MDP $\M$, whether the sequence $\Pre^n_{\M}(T) \neq \emptyset$ 
for all $n \geq 0$ is PSPACE-complete.

The reduction is as follows (see also \figurename~\ref{fig:lim-sure-reduction}). 
Given an MDP $\M=\tuple{Q, \Act,\delta}$ and a singleton $T \subseteq Q$,
we construct an MDP $\N =\tuple{Q', \Act',\delta'}$ with state space $Q' = Q \uplus \{q_{\init}\}$ such that 
$\Pre^n_{\M}(T) \neq \emptyset$ for all $n \geq 0$ if and only if 
$q_{\init}$ is limit-sure eventually synchronizing in $T$.
The MDP $\N$ is essentially a copy of $\M$ with alphabet $\Act \uplus \{\sharp\}$ and 
the transition function on action $\sharp$ is the uniform 
distribution on $Q$ from $q_{\init}$, and the Dirac distribution on 
$q_{\init}$ from the other states $q \in Q$. There are self-loops on $q_{\init}$ 
for all other actions $a \in \Act$. Formally, the
transition function $\delta'$ is defined as follows, for all $q \in Q$:

\begin{itemize} 
\item $\delta'(q,a) = \delta(q,a)$ for all $a \in \Act$ (copy of $\M$),
      and $\delta'(q,\sharp)(q_{\init}) = 1$;
\item $\delta'(q_{\init},a)(q_{\init}) = 1$ for all $a \in \Act$,
      and $\delta'(q_{\init}, \sharp)(q) = \frac{1}{\abs{Q}}$.
\end{itemize}

We establish the correctness of the reduction as follows.
For the first direction, assume that 
$\Pre^n_{\M}(T) \neq \emptyset$ for all~$n \geq 0$. 
Then since $\N$ embeds a copy of $\M$ it follows that 
$\Pre^n_{\N}(T) \neq \emptyset$ for all~$n \geq 0$ and
there exist numbers $k_0,r \leq 2^{\abs{Q}}$ 
such that $\Pre^{k_0 + r}_{\N}(T) = \Pre^{k_0}_{\N}(T) \neq \emptyset$.
Using Lemma~\ref{lem:lse-pre} with $k = k_0$ and $R = \Pre^{k_0}_{\N}(T)$ 
(and $U = Z = Q'$ is the trivial support), it is sufficient to
prove that $q_{\init} \in \winlim{event}(R)$ to get $q_{\init} \in \winlim{event}(T)$
(in $\N$).
We show the stronger statement that $q_{\init}$ is actually almost-sure 
eventually synchronizing in $R$ with the pure strategy $\alpha$ defined as follows,
for all play prefix $\rho$ (let $m = \abs{\rho} \!\!\mod r$):

\begin{itemize} 
\item if $\Last(\rho) = q_{\init}$, then $\alpha(\rho) = \sharp$;
\item if $\Last(\rho) = q \in Q$, then 
	\begin{itemize} 
	\item if $q \in \Pre^{r-m}_{\N}(R)$, then $\alpha(\rho)$ plays a $R$-safe action at position $r-m$;
	\item otherwise, $\alpha(\rho) = \sharp$.
	\end{itemize}
\end{itemize}

The strategy $\alpha$ ensures that the probability mass that is not (yet) in 
the sequence of predecessors $\Pre^n_{\N}(R)$ goes to $q_{\init}$, where by playing $\sharp$
at least a fraction $\frac{1}{\abs{Q}}$ of it would reach the sequence of predecessors
(at a synchronized position). It follows that after $2i$ steps, the probability
mass in $q_{\init}$ is $(1 - \frac{1}{\abs{Q}})^i$ and the probability mass
in the sequence of predecessors is $1 - (1 - \frac{1}{\abs{Q}})^i$. For $i \to \infty$,
the probability in the sequence of predecessors tends to $1$ and since 
$\Pre^n_{\N}(R) = R$ for all positions $n$ that are a multiple of $r$,
we get $\sup_{n}\M^{\alpha}_{n}(R) = 1$ and $q_{\init} \in \winas{event}(R)$.

For the converse direction, assume that $q_{\init} \in \winlim{event}(T)$
is limit-sure eventually synchronizing in $T$. By Lemma~\ref{lem:lse-pre},
either $(1)$ $q_{\init}$ is limit-sure eventually synchronizing in $\Pre^n_{\N}(T)$
for all $n \geq 0$, and then it follows that $\Pre^n_{\N}(T) \neq \emptyset$ 
for all $n \geq 0$, or $(2)$ $q_{\init}$ is sure eventually synchronizing in $T$,
and then since only the action $\sharp$ leaves the state $q_{\init}$ (and $\post(q_{\init},\sharp) = Q$), 
the characterization of Lemma~\ref{lem:sure-ss-pre} shows that 
$Q \subseteq \Pre^k_{\N}(T)$ for some $k \geq 0$, and since 
$Q \subseteq \Pre_{\N}(Q)$ and $\Pre_{\N}(\cdot)$ is a monotone operator, 
it follows that $Q \subseteq \Pre^n_{\N}(T)$ for all $n \geq k$ and thus 
$\Pre^n_{\N}(T) \neq \emptyset$ for all $n \geq 0$.
We conclude the proof by noting that $\Pre^n_{\M}(T) = \Pre^n_{\N}(T) \cap Q$
and therefore $\Pre^n_{\M}(T) \neq \emptyset$ for all $n \geq 0$.

\qed
\end{proof}

The example in the proof of Lemma~\ref{lem:inf-mmeory-almost-event} can be used
to show that the memory needed by a family of strategies to win 
limit-sure eventually synchronizing objective (in target $T= \{q_2\}$) 
is unbounded.

The following theorem summarizes the results for limit-sure eventually synchronizing. 

\begin{theorem}\label{theo:limit-sure-eventually}
For limit-sure eventually synchronizing (with or without exact support) in MDPs:

\begin{enumerate}
\item (Complexity). The membership problem is PSPACE-complete.

\item (Memory). Unbounded memory is required for both pure 
and randomized strategies, and pure strategies are sufficient.
\end{enumerate}

\end{theorem}


\paragraph{{\bf Acknowledgment}}
We are grateful to Winfried Just and German A. Enciso for helpful discussions
on Boolean networks and for the gadget in the proof of Lemma~\ref{lem:universal-finiteness-pspace-hard}.

\bibliographystyle{plain}
\bibliography{biblio} 

\begin{thebibliography}{10}

\bibitem{AAGT12}
M.~Agrawal, S.~Akshay, B.~Genest, and P.~S. Thiagarajan.
\newblock Approximate verification of the symbolic dynamics of {M}arkov chains.
\newblock In {\em Proc. of LICS}, pages 55--64. IEEE, 2012.

\bibitem{AspnesH90}
J.~Aspnes and M.~Herlihy.
\newblock Fast randomized consensus using shared memory.
\newblock {\em J.~Algorithm}, 11(3):441--461, 1990.

\bibitem{BS96}
E.~Bach and J.~Shallit.
\newblock {\em Algorithmic Number Theory, Vol. 1: Efficient Algorithms}.
\newblock MIT Press, 1996.

\bibitem{BBG08}
C.~Baier, N.~Bertrand, and M.~Gr{\"o}{\ss}er.
\newblock On decision problems for probabilistic {B}{\"u}chi automata.
\newblock In {\em Proc. of FoSSaCS}, LNCS 4962, pages 287--301. Springer, 2008.

\bibitem{BBS06}
C.~Baier, N.~Bertrand, and P.~Schnoebelen.
\newblock On computing fixpoints in well-structured regular model checking,
  with applications to lossy channel systems.
\newblock In {\em Proc. of LPAR}, volume 4246 of {\em LNCS}, pages 347--361.
  Springer, 2006.

\bibitem{BBMR08}
R.~Baldoni, F.~Bonnet, A.~Milani, and M.~Raynal.
\newblock On the solvability of anonymous partial grids exploration by mobile
  robots.
\newblock In {\em Proc. of OPODIS}, LNCS 5401, pages 428--445. Springer, 2008.

\bibitem{CKVAK11}
R.~Chadha, V.~A. Korthikanti, M.~Viswanathan, G.~Agha, and Y.~Kwon.
\newblock Model checking {MDP}s with a unique compact invariant set of
  distributions.
\newblock In {\em Proc. of QEST}, pages 121--130. IEEE Computer Society, 2011.

\bibitem{CH12}
K.~Chatterjee and T.~A. Henzinger.
\newblock A survey of stochastic $\omega$-regular games.
\newblock {\em J. Comput. Syst. Sci.}, 78(2):394--413, 2012.

\bibitem{deAlfaro97}
L.~de~Alfaro.
\newblock {\em Formal Verification of Probabilistic Systems}.
\newblock PhD thesis, Stanford University, 1997.

\bibitem{AHK07}
L.~de~Alfaro, T.~A. Henzinger, and O.~Kupferman.
\newblock Concurrent reachability games.
\newblock {\em Theor. Comput. Sci.}, 386(3):188--217, 2007.

\bibitem{DMS11b}
L.~Doyen, T.~Massart, and M.~Shirmohammadi.
\newblock Infinite synchronizing words for probabilistic automata.
\newblock In {\em Proc. of MFCS}, LNCS 6907, pages 278--289. Springer, 2011.

\bibitem{DMS11Err}
L.~Doyen, T.~Massart, and M.~Shirmohammadi.
\newblock Infinite synchronizing words for probabilistic automata ({E}rratum).
\newblock {\em CoRR}, abs/1206.0995, 2012.

\bibitem{FokkinkP06}
W.~Fokkink and J.~Pang.
\newblock Variations on {I}tai-{R}odeh leader election for anonymous rings and
  their analysis in {PRISM}.
\newblock {\em Journal of Universal Computer Science}, 12(8):981--1006, 2006.

\bibitem{GO10}
H.~Gimbert and Y.~Oualhadj.
\newblock Probabilistic automata on finite words: Decidable and undecidable
  problems.
\newblock In {\em Proc. of ICALP (2)}, LNCS 6199, pages 527--538. Springer,
  2010.

\bibitem{HMW09}
T.~A. Henzinger, M.~Mateescu, and V.~Wolf.
\newblock Sliding window abstraction for infinite {M}arkov chains.
\newblock In {\em Proc. of CAV}, LNCS 5643, pages 337--352. Springer, 2009.

\bibitem{Holzer95}
M.~Holzer.
\newblock On emptiness and counting for alternating finite automata.
\newblock In {\em Developments in Language Theory}, pages 88--97, 1995.

\bibitem{AFA1}
P.~Jancar and Z.~Sawa.
\newblock A note on emptiness for alternating finite automata with a one-letter
  alphabet.
\newblock {\em Inf. Process. Lett.}, 104(5):164--167, 2007.

\bibitem{KVAK10}
V.~A. Korthikanti, M.~Viswanathan, G.~Agha, and Y.~Kwon.
\newblock Reasoning about {MDP}s as transformers of probability distributions.
\newblock In {\em Proc. of QEST}, pages 199--208. IEEE Computer Society, 2010.

\bibitem{PK11}
M.A. Pinsky and S.~Karlin.
\newblock {\em An Introduction to Stochastic Modeling}.
\newblock Academic Press. Academic Press, 2011.

\bibitem{Vardi-focs85}
M.~Y. Vardi.
\newblock Automatic verification of probabilistic concurrent finite-state
  programs.
\newblock In {\em Proc. of FOCS}, pages 327--338, 1985.

\bibitem{Volkov08}
M.~V. Volkov.
\newblock Synchronizing automata and the {C}erny conjecture.
\newblock In {\em Proc. of LATA}, LNCS 5196, pages 11--27. Springer, 2008.

\end{thebibliography}



\end{document}